\documentclass[lettersize,journal,article]{IEEEtran}
\usepackage{tikz}
\usepackage{array}
\usepackage{times}
\usepackage{subfig}
\usepackage{helvet}
\usepackage{subfig}
\usepackage{amssymb}
\usepackage{amsthm}
\usepackage{amsmath}
\usepackage{amssymb}
\usepackage{courier}
\usepackage{multirow}
\usepackage{mathrsfs}
\usepackage{graphicx}
\usepackage{enumitem}
\usepackage{graphicx}
\usepackage{blindtext}
\usepackage{algorithm}
\usepackage{algorithmic}
\usepackage{lineno,hyperref}
\usepackage[marginal]{footmisc}
\usepackage{amsthm}
\usepackage[utf8]{inputenc}
\usepackage[english]{babel}
\usepackage{epstopdf}
\usepackage{array}
\usepackage{multirow}
\usepackage{booktabs}
\usepackage{dsfont}
\usepackage{comment}
\usepackage{endnotes}
\usepackage{makecell}
\usepackage{listings}
\usepackage{xcolor}
\usepackage{bm}
\usepackage[normalem]{ulem}
\useunder{\uline}{\ul}{}
\newtheorem{theorem}{Theorem}
\newtheorem{definition}{Definition}

\begin{document}

\title{Dual Test-time Training for Out-of-distribution Recommender System}

\author{Xihong~Yang$^{\ast}$,~Yiqi~Wang$^{\ast}$,~Jin~Chen,~Wenqi~Fan,~Xiangyu~Zhao,~En~Zhu$^{\dagger}$,\\
~Xinwang~Liu$^{\dagger}$,~\IEEEmembership{Senior~Member,~IEEE}, Defu Lian$^{\dagger}$
\thanks{X. Yang, Y. Wang, X. Liu and E. Zhu are with School of Computer, National University of Defense Technology, Changsha, 410073, China. (E-mail: \{yangxihong,\,yiq,\, xinwangliu,\, enzhu\} @nudt.edu.cn)}
\thanks{J. Chen is with School of Business and Management of the Hong Kong University of Science and Technology. (E-mail: jinchen@ust.hk)}
\thanks{W. Fan is with the Department of Computing (COMP) and Department of Management and Markering (MM), The Hong Kong Polytechnic University. (E-mail: wenqifan03@gmail.com)}
\thanks{X. Zhao is with Department of Data Science of City University of Hong Kong. (Email: xianzhao@cityu.edu.hk)}
\thanks{D. Lian is with the Anhui Province Key Laboratory of Big Data Analysis and Application, School of Computer Science and Technology, University of Science and Technology of China, Hefei, Anhui 230000, China. (Email: liandefu@ustc.edu.cn).}
\thanks{$^{\ast}$: Equal contribution. $^{\dagger}$: Corresponding author.}
}

\markboth{IEEE Transactions on Knowledge and Data Engineering}%
{X. Yang \MakeLowercase{\textit{et al.}}: Dual Test-time Training for Out-of-distribution Recommender System}


\maketitle

\begin{abstract}
Deep learning has been widely applied in recommender systems, which has recently achieved revolutionary progress. However, most existing learning-based methods assume that the user and item distributions remain unchanged between the training phase and the test phase. However, the distribution of user and item features can naturally shift in real-world scenarios, potentially resulting in a substantial decrease in recommendation performance. This phenomenon can be formulated as an Out-Of-Distribution (OOD) recommendation problem. To address this challenge, we propose a novel \textbf{D}ual \textbf{T}est-\textbf{T}ime-\textbf{T}raining framework for \textbf{O}OD \textbf{R}ecommendation, termed \textbf{DT3OR}. In DT3OR, we incorporate a model adaptation mechanism during the test-time phase to carefully update the recommendation model, allowing the model to adapt specially to the shifting user and item features. To be specific, we propose a self-distillation task and a contrastive task to assist the model learning both the user's invariant interest preferences and the variant user/item characteristics during the test-time phase, thus facilitating a smooth adaptation to the shifting features. Furthermore, we provide theoretical analysis to support the rationale behind our dual test-time training framework. To the best of our knowledge, this paper is the first work to address OOD recommendation via a test-time-training strategy. We conduct experiments on five datasets with various backbones. Comprehensive experimental results have demonstrated the effectiveness of DT3OR compared to other state-of-the-art baselines.
\end{abstract}


\section{Introduction}
\IEEEPARstart{R}{ecommender} systems play a crucial role in alleviating the information overload on social media platforms by providing personalized information filtering. In recent years, a plethora of recommendation algorithms have been proposed, including collaborative filtering~\cite{chen_cf, yq1, yq_cf,yin2024dataset,yqadaptive, liuyue_ELCRec, liuyue_ITR, lxy_llmrec, lxy_llmrec2}, graph-based recommendation~\cite{LightGCN,RMLRec, Rensslrec,DCCF,Fan_graph}, cross-domain recommendation~\cite{zhao_tkde, zhao2023cross,yin2023apgl4sr,yang2024hyperbolic} and etc.

\begin{figure}
\centering
\scalebox{0.3}{
\includegraphics{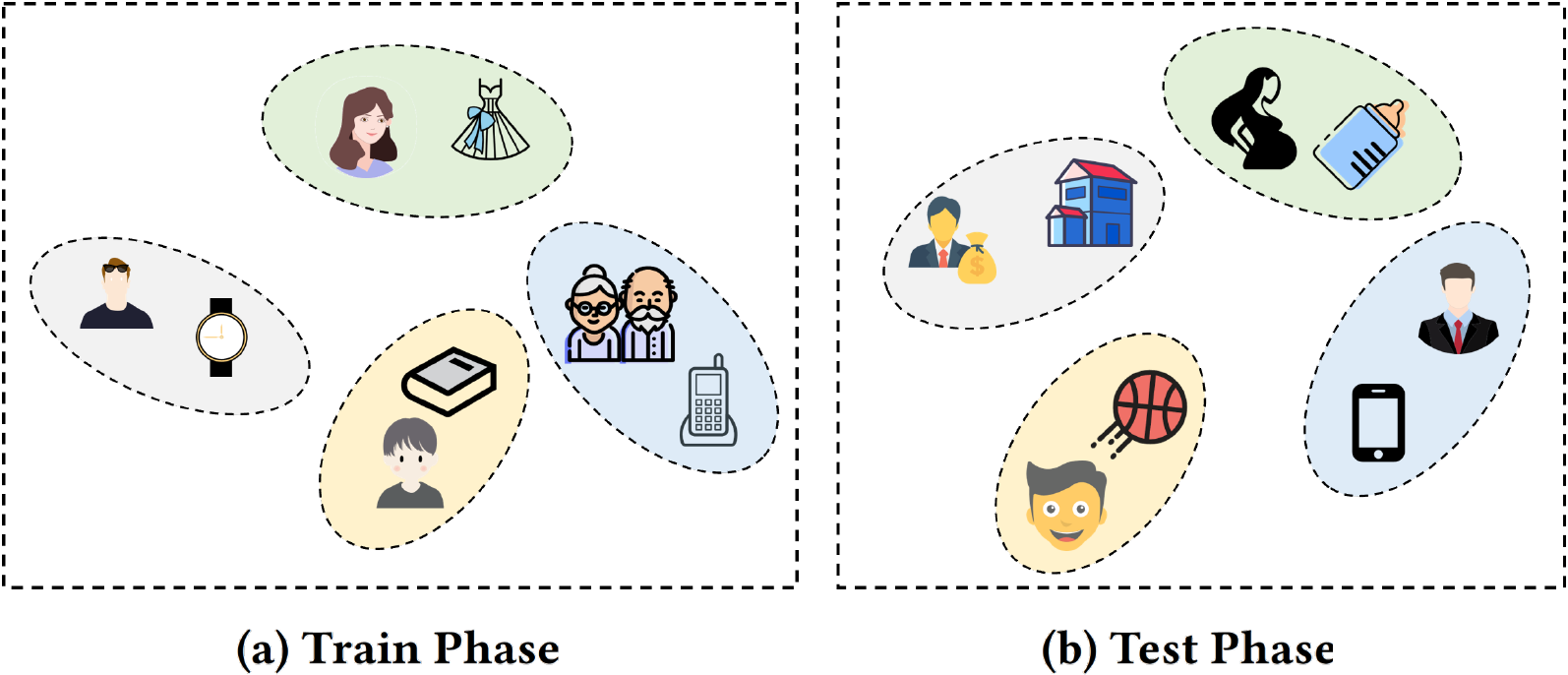}}
\caption{Illustration of OOD recommendation. The distribution shift of user and item features in test phase, leading to decline in recommendation performance.}
\label{motivation}
\end{figure}

Although promising performance has been achieved, most existing recommendation algorithms assume that the user and item distributions remain unchanged between the training phase and the test phase, adhering to the assumption of independent and identically distributed (I.I.D.) variables. However, these distributions often undergo shifts in many real-world scenarios, violating this assumption. We illustrate the scenario of distribution drift between the training set and the test set in Fig.~\ref{motivation}. The shifts primarily manifest in two aspects: user feature shift and item feature shift. {User feature shift refers specifically to the evolution of a user’s preferences over time as their individual characteristics change. Specifically, we consider that as individual characteristics evolve over time, there will be corresponding changes in user preferences. For example, as a boy matures, his interests may shift from comic books to sports. Similarly, as he grows older and experiences an increase in income, purchasing a house may become a significant focus during adulthood. In contrast, item feature shift pertains to changes in preference relationships caused by the replacement or updating of item characteristics. For instance, items frequently undergo modifications that alter their features, which may result in outdated user-item interactions and inappropriate recommendations. Consider a restaurant whose signature dishes vary with the seasons: in spring, fresh seasonal vegetables may dominate the menu; in summer, refreshing beverages; in autumn, an assortment of fruits; and in winter, hot pot dishes. Therefore, it is necessary to adapt the recommendation strategy to align with the changing characteristics of the restaurant (item). We demonstrate the user/item feature shift in Fig.~\ref{feature_shift}.} These shifts significant impact on recommendation performance~\cite{cor,causpref, AdvInfoNCE, AdvDrop}. Consequently, it is crucial to consider out-of-distribution (OOD) recommendation to address this issue.

\begin{figure}
\centering
\scalebox{0.45}{
\includegraphics{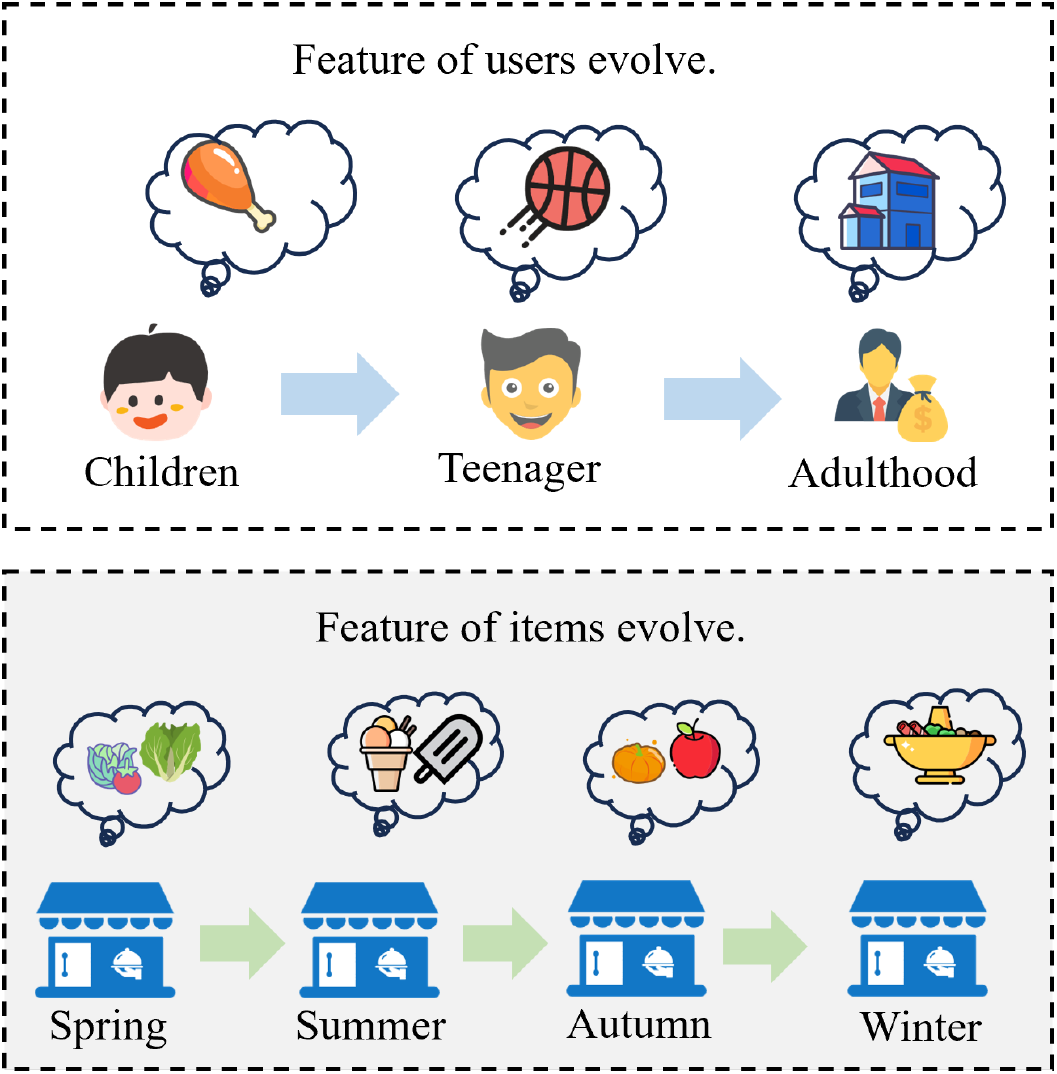}}
\caption{{Illustration of user/item feature shift.}}
\label{feature_shift}
\end{figure}

The out-of-distribution (OOD) problem has been extensively studied in various fields, including image classification~\cite{ood_classifi} and detection~\cite{ood_detection}. However, it has received limited attention in the context of recommendation systems. To address this issue, disentangled recommendation approaches~\cite{dis_rec, dis_rec1} aim to learn factorized representations of user preferences, which can enhance the model robustness against distribution shifts. Causal-based methods~\cite{cor,causpref} use causal learning to solve OOD problems. However, these methods often require interventions during the training process, which makes them less practical when only a pre-trained model is available. Another approach to address OOD recommendation is model retraining~\cite{retraining}, which involves adapting the model to the OOD environment. However, this approach often requires a substantial computational overhead~\cite{cor}. In contrast, our study aims to investigate practical approaches that can readily improve the generalization and robustness for a wide range of pre-trained models in different scenarios. Test-time training (TTT), a widely adopted strategy for solving OOD problems in various domains~\cite{ttt++}, serves as the inspiration for our proposed solution. The core concept of test-time training is to utilize auxiliary unsupervised or self-supervised learning (SSL) tasks to provide insights into the underlying characteristics of a test sample. This information is then used to fine-tune or calibrate the trained model specifically for that test sample. Since each sample has distinct properties, test-time training has significant potential to enhance the generalization capabilities of trained models.

To address the OOD recommendation challenge, we present a novel dual test-time-training strategy called DT3OR. This strategy leverages test-time training to adapt the model during the recommendation phase and improve its performance in OOD scenarios. In DT3OR, we introduce a model adaptation mechanism during the test-time phase to effectively update the recommendation model, thus allowing the model to adapt to the shift data. This mechanism consists of two key components to better understand the invariant preferences among users and the variant user/item features with shift data: the self-distillation task and the contrastive task. The self-distillation task aims to minimize the distance between user interest centers with the same preference, thereby improving the uniformity of user interest representations in the latent space. By encouraging the user interest centers to be more tightly clustered, the model enhances its proficiency in capturing the inherent patterns in user preference. Additionally, we incorporate a contrastive task that captures the correlations among users with similar preferences. This task involves selecting high-confidence samples from the same user interest center to establish meaningful relationships between users.

In summary, the main contributions of this work are summarized as follows:
\begin{itemize}

    \item In this work, we design test-time training strategies that enable the model to adapt to shifts in feature distributions, addressing the out-of-distribution (OOD) recommendation challenge.
            
    \item In our paper, two self-supervised strategies are introduced to capture users’ invariant preferences, i.e., self-distillation task and contrastive task. To the best of our knowledge, this is the first study to leverage test-time training for addressing the OOD recommendation problem.
    
    \item We provide a theoretical analysis of the effectiveness of the proposed test-time training strategy design.
      
    \item A comprehensive set of comparative experiments on five datasets confirms the effectiveness of our proposed approach relative to existing state-of-the-art methods.

\end{itemize}

\begin{figure}
\centering
\scalebox{0.6}{
\includegraphics{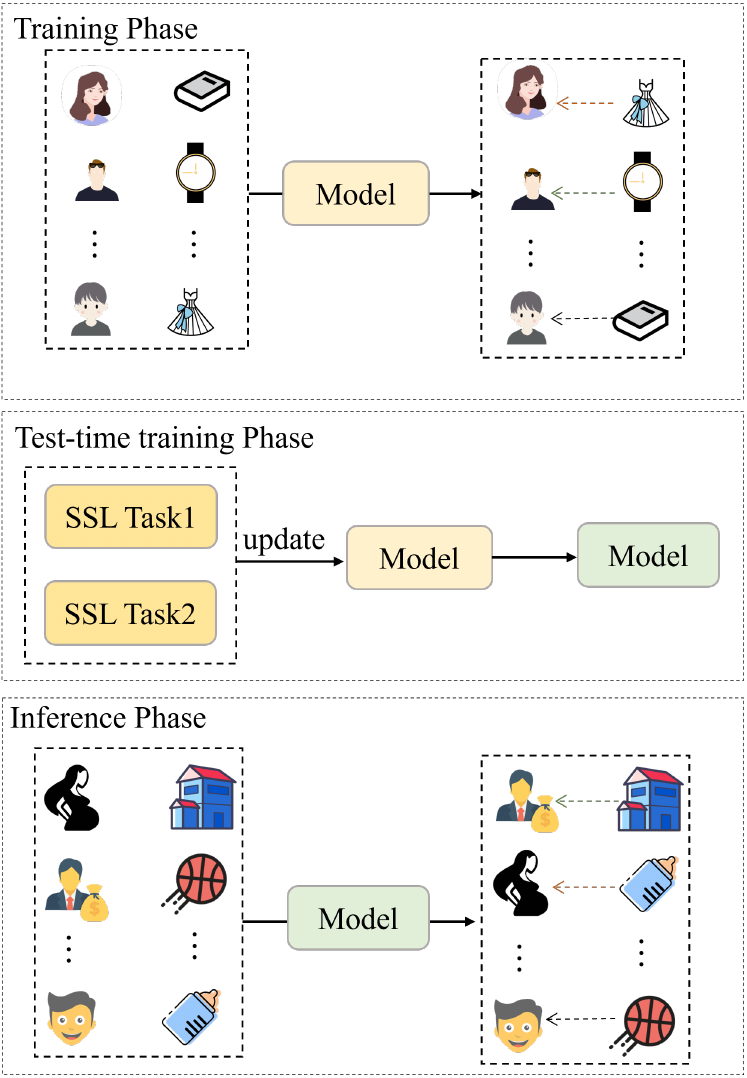}}
\caption{The overall framework for our proposed DT3OR.}
\label{overall}
\end{figure}

\begin{figure*}
\centering
\scalebox{0.6}{
\includegraphics{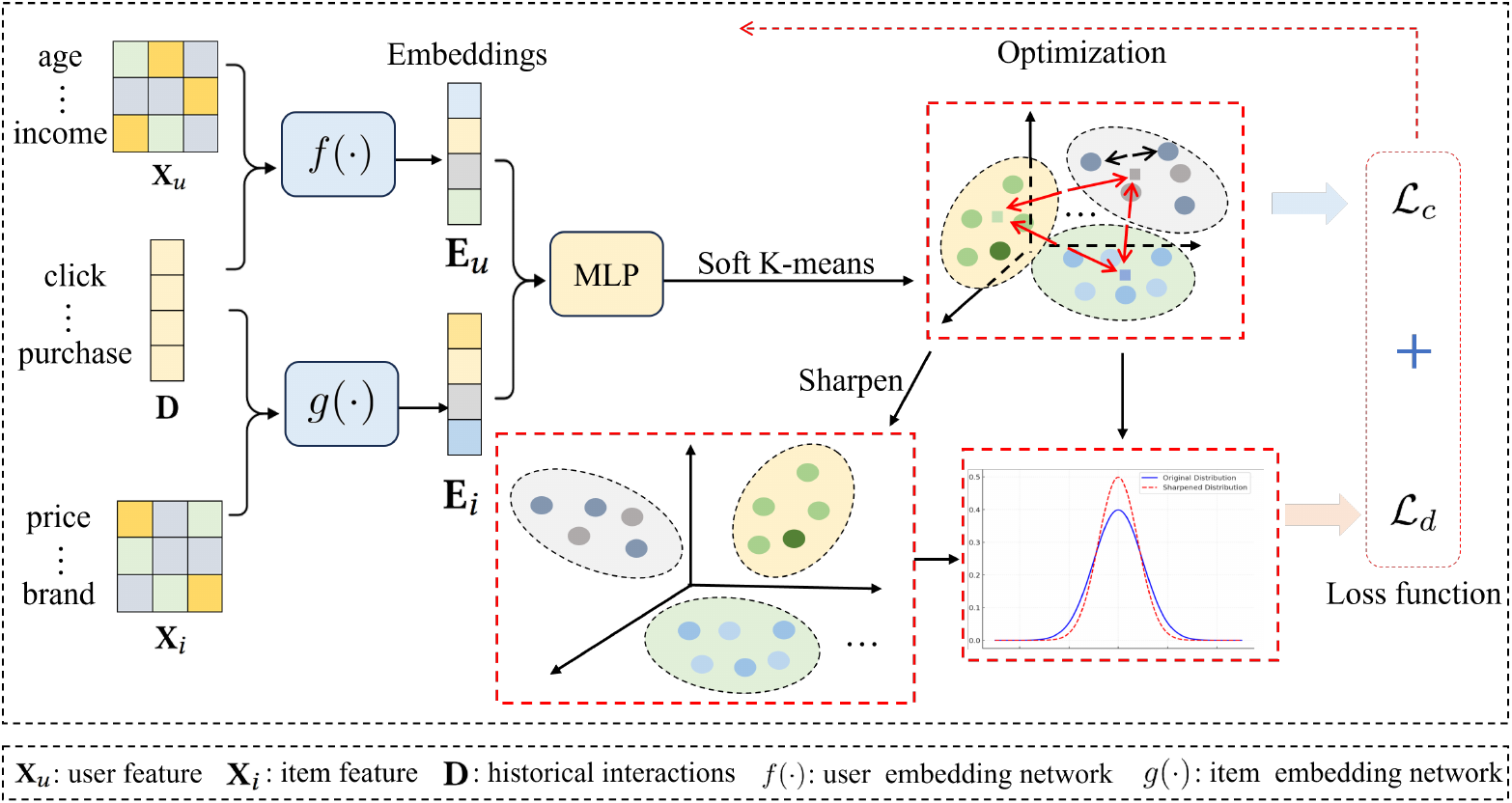}}
\caption{Illustration of the self-distillation task designing. In our proposed method, we initially obtain the user embedding $\textbf{E}_u$ and item embedding $\textbf{E}_i$ by encoding the user/item features and historical interactions with the training network. After this, we carry out K-means on the combined embedding to yield the clustering results. Subsequently, with the sharpen function, we implement self-distillation to enhance the user interest centers, thus improving the generalization of model to deal with shift features.}
\label{SSL1}
\end{figure*}

\section{Problem Definition}
\textbf{OOD Recommendation Scenario.}
In this work, our primary focus is on addressing the out-of-distribution (OOD) problem in the recommendation task. We consider a training dataset denoted as $\mathcal{D}_{Tr}$, and a test dataset denoted as $\mathcal{D}_{Te}$. The corresponding set of interactions are denoted as $\mathcal{Y}_{Tr}$ and $\mathcal{Y}_{Te}$. Moreover, in an OOD scenario, $\mathcal{D}_{Tr}$ and $\mathcal{D}_{Te}$ are disjoint, and they are drawn from the different distribution.

\begin{definition} [Test-time training for OOD Recommendation (T3OR).] T3OR requires the learning of a self-supervised learning task denoted as $h(\cdot)$ during the test phase. By leveraging a pre-trained network $f_{\theta}$, T3OR aims to improve the performance on the test data $\mathcal{D}_{Te}$:
\begin{equation}
\begin{aligned}
&\underset{h}{argmin} \mathcal{L}\left ( f_{\theta^*}\left ( h\left ( \mathcal{D}_{Te}\right ),\mathcal{Y}_{Te}  \right )  \right )\\
&s.t.~~h\left ( \mathcal{D}_{Te}  \right ) \in \mathcal{P}\left ( \mathcal{D}_{Te} \right )\\ 
&with~\theta^* =  \underset{\theta}{argmin} \mathcal{L}\left ( f_{\theta}\left ( \mathcal{D}_{Tr}\right ) ,\mathcal{Y}_{Tr}  \right ) 
\end{aligned}
\end{equation}
where $\mathcal{L}$ is the loss function to measure the performance of downstream task, i.e., recommendation. $\mathcal{P}\left ( \mathcal{D}_{Te} \right )$ is the set of self-supervised tasks.
\end{definition}

\section{Preliminary}
 {This work proposes test-time training strategies to assist the model in adapting to out-of-distribution scenarios during the test phase.
Let $\textbf{E}_u$ and $\textbf{E}_i$ represent the embeddings of users and items, respectively, obtained from the user embedding extraction network $f(\cdot)$ and the item embedding extraction network $g(\cdot)$ in the latent space. $\textbf{E}$ denotes the fused embeddings of users and items. Additionally, $\textbf{D}$ represents the historical interactions. We use $\textbf{Q}$ and $\textbf{P}$ to denote the original and sharpened distributions of pseudo-labels generated through the clustering operation. $\textbf{R}$ denotes the disjoint clusters, while $\textbf{C}$ represents the user preference clusters. The basic notations are summarized in Tab.~\ref{notation}.}

\section{Methodology}
We design a novel framework termed DT3OR for recommendation under an OOD scenario in this section. {The framework of DT3OR is shown in Fig. \ref{overall}.} The primary obstacle in optimizing the DT3OR problem lies in constructing a surrogate loss that can guide the learning process, particularly considering the lack of ground-truth interactions for the test data. To be specific, we first delve into the designed self-supervised learning tasks, which include the self-distillation task and the contrastive task. Then, we give the definition of loss function in DT3OR.

\subsection{Self-supervised Learning Tasks Design}

Our proposed DT3OR is designed to improve model generalizability and robustness to handle OOD recommendations by learning from test data. Ideally, if we had the ground-truth labels for the test data, we could refine the model to minimize the cross-entropy loss on the test data. However, due to the absence of such label information at the test-time, we are motivated to explore suitable surrogate loss functions to steer the model without the presence of labeled data. In the unsupervised scenario, self-supervised learning (SSL) techniques~\cite{DealMVC, CodingNet,xihong} have shown great potential, laying the groundwork for self-supervision during test-time training. The effectiveness of test-time training is heavily dependent on the quality SSL task~\cite{ttt++}. However, SSL tasks commonly used in image-related tasks, such as rotation invariance, may not be directly applicable to recommendation data due to its distinct characteristics. Therefore, it is vital to build up a reasonable SSL task based on the characteristics of recommendation scenarios. {In recommendation, user preferences often exhibit clustering and homogeneity, where users with similar purchasing behaviors tend to have comparable profiles and are likely to be interested in similar new items.} Inspired by this intuition, we designed a dual self-supervised task to assist the model to adapt to the shift distribution in test-time phase, achieving improved recommendation performance on the test set. In following, we design two SSL tasks that are suitable for recommendation tasks. {The framework for designed SSL tasks is illustrated in Fig.~\ref{SSL1}.}

\begin{table}[]
\centering
\caption{{Notation Summary.}}
\begin{tabular}{c|c}
\hline
\textbf{{Notation}} & \textbf{{Meaning}}                 \\ \hline
{$\textbf{X}_u \in \mathbb{R}^{N_U \times d'}$}                & {Features of users}                \\
{$\textbf{X}_i \in \mathbb{R}^{N_I \times d'}$}               & {Features of items}                \\
{$\textbf{D} \in \mathbb{R}^{N_D \times d}$}                 & {Historical interactions}          \\
{$\textbf{E}_u \in \mathbb{R}^{N_D \times d}$}                & {Embeddings of users}             \\
{$\textbf{E}_i \in \mathbb{R}^{N_D \times d}$}                & {Embeddings of items}              \\
{$\textbf{E} \in \mathbb{R}^{N_D \times 2d}$}                & {Fused embeddings of items and users.}            \\
{$\textbf{Q} \in \mathbb{R}^{N_D \times K}$}                 & {Distribution of pseudo labels}    \\
{$\textbf{P} \in \mathbb{R}^{N_D \times K}$}                 & {Sharpened distribution}           \\
{$\textbf{R}_i \in \mathbb{R}^{N_D}$}                & {Clusters of user preference}      \\
{$\textbf{C}_i \in \mathbb{R}^{K}$}                 & {Cluster centers}                  \\
{$f(\cdot)$}                & {user embedding extracted network} \\
{$g(\cdot)$}                  & {item embedding extracted network} \\ \hline
\end{tabular}
\label{notation}
\end{table}

\subsubsection{Self-distillation Task Design}
In this subsection, we propose a test-time self-distillation task to make better leverage of the test data. To be specific, we first extract the embeddings of user/item features and historical interactions {by the pretrained user embedding extracted network $f(\cdot)$ and item embedding extracted network $g(\cdot)$ in the training phase:}
\begin{equation}
\begin{aligned}
\textbf{E}_u &= f(\textbf{X}_u, \textbf{D})\\
\textbf{E}_i &= g(\textbf{X}_i, \textbf{D}),
\end{aligned}
\label{user_item_embed}
\end{equation}
where $\textbf{X}_u$ and $\textbf{X}_i$ denote the user features and item features, respectively. $\textbf{D}$ is the historical interactions. $\textbf{E}_u$ and $\textbf{E}_i$ capture the preferences of users towards items. In a recommendation system, historical interactions and user features are valuable indicators of user preferences~\cite{VAERec}. {Users' purchasing behaviors and their interactions with items offer meaningful insights into their underlying preferences.}

Then, we fuse the user embeddings $\textbf{E}_u$ and item embeddings $\textbf{E}_i$ by:
\begin{equation}
\begin{aligned}
\textbf{E} &= \gamma(\textbf{E}_u \oplus \textbf{E}_i),
\end{aligned}
\label{MLP_embed}
\end{equation}
where $\oplus$ denotes the the concatenation operation, $\gamma(\cdot)$ denotes a multi-layer perceptron (MLP). After that, we implement soft K-means algorithm on $\textbf{E}$. {The prediction probability distribution of the fused embedding $\textbf{E}$ could be obtained by:}
\begin{equation}
\begin{aligned}
q &= \varphi(\textbf{E}).
\end{aligned}
\label{clustering}
\end{equation}

{Here, $\varphi(\cdot)$ represents the soft K-means algorithm, which calculates the assignment probability of each sample to clusters based on its distance from the cluster centers in the latent space. The variable $q$ can be interpreted as the probability of assigning a sample to a specific cluster, indicating a soft assignment.} We can consider the matrix $\textbf{Q}=[q]$ as the distribution of assignments for all samples.

After obtaining the distribution of the clustering result $\textbf{Q}$, we utilize self-distillation to optimize the network by leveraging a lower entropy assignments. In particular, we employ a sharpening function to diminish the entropy of the clustering result distribution $\textbf{Q}$. We determine the target distribution $\textbf{P}$ through the following calculation:

\begin{equation}
\begin{aligned}
\textbf{P}=S(q,T)_{i} &= \frac{q_i^{\frac{1}{T}}}{\sum_{j}q_j^{\frac{1}{T}}}.
\end{aligned}
\label{sharpen_func}
\end{equation}
where $S(\cdot)$ represents the sharpening function, $T$ is a temperature parameter, and $q$ denotes input probability distribution, which corresponds to the clustering result distribution in this paper. As the temperature $T$ approaches zero, the output of $S(q,T)$ tends to approximate a Dirac distribution (one-hot encoding). A lower temperature encourages the generation of a lower-entropy distribution. Consequently, we can formulate the self-distillation objective function as follows:
\begin{equation}
\begin{aligned}
\mathcal{L}_d &= KL(\textbf{P}||\textbf{Q})=\sum_ip_ilog\frac{p_i}{q_i}.
\end{aligned}
\label{kl_loss}
\end{equation}

By minimizing the KL divergence loss between the $\textbf{P}$ and $\textbf{Q}$ distributions, the target distribution $\textbf{P}$ aids the model in acquiring better embeddings, given the low entropy of $\textbf{P}$. This process can be regarded as a self-distillation mechanism, considering that $\textbf{P}$ is derived from $\textbf{Q}$. Furthermore, the $\textbf{P}$ distribution guides the update of the $\textbf{Q}$ distribution in a supervisory capacity.

Through clustering the fused embeddings of users and items, we obtain interest centers that represent the preferences of different types of users, where each type of user shows a preference for a specific type of item. The self-distillation mechanism designed in this paper enhances the prominence of these interest centers, allowing for better understanding and modeling of user preferences.

\begin{figure}
\centering
\scalebox{0.3}{
\includegraphics{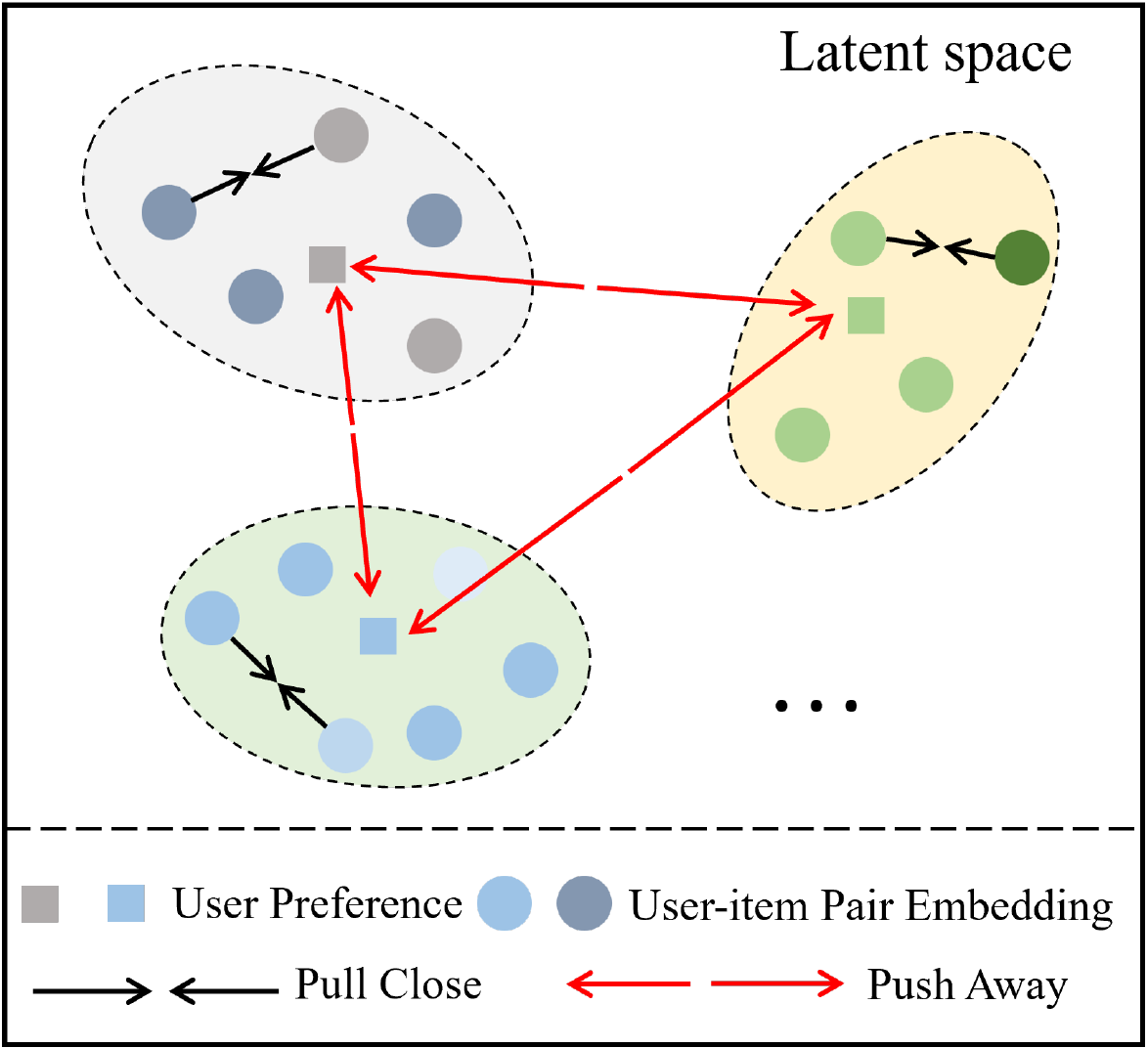}}
\caption{Illustration of the contrastive task designing. With the high confidence results of clustering, we select the samples in the same interest center as positive sample, while regarding the different high confidence interest centers as negative samples.}
\label{SSL2}
\end{figure}

\subsubsection{Contrastive Task Design}
{Contrastive learning has surfaced as a significant self-supervised framework, demonstrating promising performance across various tasks~\cite{SIMCLR,BYOL,GTrans}.} Data augmentation is pivotal in contrastive learning. However, inappropriate data augmentation can lead to semantic drift, thereby degrading the quality of positive samples~\cite{SCAGC}. {Additionally, existing strategies for constructing negative samples often consider all non-positive samples as negative, which can inadvertently introduce false-negative samples~\cite{GDCL}. Therefore, we aim to design a contrastive task with more discriminative positive samples and reliable negative samples. {The specific process is shown in Fig.~\ref{SSL2}} The designed contrastive task contains two steps.}

Firstly, we construct the positive sample pairs. Based on the clustering distribution $q$ with Eq.~\ref{clustering}, we identify the high confidence sample indexes as:
\begin{equation}
\begin{aligned}
h = \text{Top}(q), 
\end{aligned}
\label{high_index}
\end{equation}
where $\text{Top}(\cdot)$ is used to select the indices of the top $\tau$ high-confidence samples in the clustering distribution $q$. According to these indices, we denote these high-confidence embeddings as $\textbf{E}_h$. Subsequently, based on the pseudo labels derived from the clustering operation in Eq.~\eqref{clustering}, we partition $\textbf{E}_{h}$ into $K$ disjoint clusters $\textbf{R}_i$, which can be represented as:

\begin{equation}
\begin{aligned}
\textbf{R}_i = \text{Div}({\textbf{E}_{h}}), 
\end{aligned}
\label{cluster}
\end{equation}
where $\text{Div}(\cdot)$ is the division function. $i \in [1,K]$. By this way, the high-confidence embeddings could be divided into different clusters, which denotes different preferences in user-item pairwise level. We consider the samples in the same clusters as positive samples. In other words, the same preference relationship (user-like-item) is considered a positive samples.

After that, we calculate the centers of high confidence embeddings to construct negative samples:
\begin{equation}
\begin{aligned}
\textbf{C}_i = \frac{1}{K}\sum_{i=1}^K{\textbf{R}_i},
\end{aligned}
\label{centers}
\end{equation}
where $\textbf{C}_i$ represents different user preference centers. We consider different high-confidence centers as negative samples, as they correspond to different interest centers. In this way, we could reduce the possibility of false-negative samples.

In conclusion, utilizing high-confidence clustering information as a form of supervised information could boost the quality of both positive and negative samples, thereby resulting in enhanced performance in contrastive tasks.

\begin{algorithm}[t]
\small
\caption{\textbf{Dual Test-Time Training for OOD Recommendation }}
\label{ALGORITHM}
\flushleft{\textbf{Input}: The latest shift user/item features $\textbf{X}_u, \textbf{X}_i$, the historical interactions $\textbf{D}$}\\
\vspace{-5pt}
\flushleft{\textbf{Output}: The interaction probability $\textbf{D}'$.} 
\begin{algorithmic}[1]
\vspace{-10pt}
\FOR{$e=1$ to $epoch$}
\STATE Acquire the embedding of user and item respectively with Eq.~\eqref{user_item_embed}.
\STATE Fuse $\textbf{E}_u$ and $\textbf{E}_i$ by MLP with Eq.~\eqref{MLP_embed}.
\STATE Obtain the clustering result with K-means with Eq.~\eqref{clustering}.
\STATE Compute the self-distillation loss with Eq.~\eqref{kl_loss}.
\STATE Compute the contrastive loss with Eq.~\eqref{con_loss}.
\STATE Optimize the network $f_\theta(\cdot)$ through minimizing $\mathcal{L}$ in Eq.~\eqref{total_loss}.
\ENDFOR
\STATE Calculate the interaction probability $\textbf{D}'$. 
\end{algorithmic}
\end{algorithm}

\subsection{Objective function}\label{complex_loss}
We define the loss function of our DT3OR in OOD recommendation in this subsection. DT3OR jointly optimizes two components, i.e., the self-distillation loss $\mathcal{L}_d$ and contrastive loss $\mathcal{L}_c$.

In detail, we utilize the Mean Squared Error loss between positive samples as the positive loss:
\begin{equation}
\begin{aligned}
\mathcal{L}_p = \frac{1}{KN(N-1)}{\sum_{m=1}^K\sum_{i=1}^N\sum_{j=i+1}^N{(\textbf{R}^{m}_{i} - \textbf{R}^{m}_{j})^2}},
\end{aligned}
\label{pos}
\end{equation}
where $K$ is utilized to denote the number of interest clusters. $\textbf{R}^{m}_{i}$ and $\textbf{R}^{m}_{j}$ denotes the embeddings of different sample in the same cluster $\textbf{R}^{m}$. Besides, we use $N$ to present the number of high-confidence samples.

Moreover, we implement the negative sample loss with cosine similarity among different centers of the selected high-confidence embeddings. The initial step involves computing the cosine similarity between different centers:

\begin{equation}
\begin{aligned}
sim(\textbf{C}_i,\textbf{C}_j) = \frac{\textbf{C}_i\cdot \textbf{C}_j}{||\textbf{C}_i||_2||\textbf{C}_j||_2}.
\end{aligned}
\end{equation}

Then, we maximize the cosine similarity for different centers as:
\begin{equation}
\begin{aligned}
\mathcal{L}_n = \frac{1}{K^2-K}\sum_{i=1}^K\sum_{j=i+1}^K {{sim(\textbf{C}_i,\textbf{C}_j)}} ,
\end{aligned}
\label{neg}
\end{equation}
where $\textbf{C}_i$ and $\textbf{C}_j$ denote different high confidence center. In summary, the contrastive loss $\mathcal{L}_c$ is:
\begin{equation}
\begin{aligned}
\mathcal{L}_c = \mathcal{L}_p + \mathcal{L}_n.
\end{aligned}
\label{con_loss}
\end{equation}

The whole loss of DT3OR can be presented as follows:
\begin{equation}
\begin{aligned}
\mathcal{L} = \mathcal{L}_d + \alpha \mathcal{L}_c,
\end{aligned}
\label{total_loss}
\end{equation}
where we use $\alpha$ to represent the trade-off hyper-parameter between $\mathcal{L}_d$ and  $\mathcal{L}_c$. 

The first term in Eq.~\eqref{total_loss} promotes the distribution of user interest centers closer together, enhancing the uniformity of user interest representations in the latent space. {We first obtain the probability distribution $\textbf{P}$ of users' preference centers based on the fused user-item embeddings through a clustering operation. Then, we apply the sharpen function to generate a lower-entropy distribution $\textbf{Q}$. This distribution reflects the consistency of user preferences in the latent space, indicating that a group of users shares a similar preference for a specific category of items. By minimizing the KL divergence between $\textbf{P}$ and $\textbf{Q}$, we bring the distributions of user interest centers closer together.} The second term captures the correlations among users with similar preferences. Here, we use ${B}$ to represent the batch size. Furthermore, the embeddings' dimension is set to $d$, and $N$ as the number of high-confidence samples. Consequently, we could obtain time complexity of the self-distillation loss and contrastive loss is $\mathcal{O}(Bd)$ and $\mathcal{O}(N^2d)$, respectively. Therefore, the total time complexity of the loss is $\mathcal{O}(Bd+N^2d)$. The space complexity of our loss is $\mathcal{O}(B+N^2)$. The comprehensive learning process of DT3OR is outlined in Algorithm.~\ref{ALGORITHM}.

\section{Theoretical analysis}
We explore the rationality of test-time training for recommendation from a theoretical perspective in this section. For the sake of convenience, we establish the following notation. Based on Eq.~\eqref{ce_loss}, let $\mathcal{L}(u,i,y;\theta)$ denote the cross-entropy loss between the ground-truth $y_i$ and the prediction $\hat{y}_i$. Let $\mathcal{L}_{SUP}(u,i,y;\theta)$ denote a supervised task loss on a instance $(u,i,y)$ with parameters $\theta$, e.g., recommendation task, and $\mathcal{L}_{SSL}(u,i;{\bm \theta})$ denotes a self-supervised task loss on a instance $(u,i))$ also with parameters $\bm \theta$, e.g., the designing self-supervised tasks in this paper. Here, we use the commonly recommendation loss function, i.e., CE loss, to demonstrate, which can be presented as:
\begin{equation}
 \mathcal{L}_{CE}=-\frac{1}{N}\sum_{i=1}^N y_{i}\text{log}(\hat{y}_i),
\label{ce_loss}
\end{equation}
where $y_i$ denotes the ground-truth. $\hat{y}_i$ is the prediction of the model embedding, which can be presented as:
\begin{equation}
\begin{aligned}
    \hat{y}_i &= softmax(f(u,i))\\
    &=softmax(g_i)\\
    &=\frac{e^{g_i}}{\sum_{j=1}^Ne^{g_j}}
\end{aligned}
\label{softmax}
\end{equation}
where $f(\cdot)$ is the recommendation model, $y_i$ is defined as the variable indicating whether the user has interacted with an item, while $\hat{y}_i$ represents the predicted probability of interaction between the user and the item by the model.

\begin{theorem}
$\mathcal{L}(u,i,y;\theta)$ is convex and $\beta$-smooth with respect to ${\theta}$, and the magnitude of the gradient $|\nabla{\bm \theta}\mathcal{L}(u,i,y;\theta) |$ is bounded by a positive constant $B^{'}$ for all ${\bm \theta}$.
\label{the_1}
\end{theorem}

\begin{proof}With the inner product of embeddings $g = f(u, i)$ and based on Eq.~\eqref{ce_loss}, we could have the following formula:

\begin{equation}
\begin{aligned}
\mathcal{L}(y;g) &= -\frac{1}{N}\sum_{i=1}^N(y_i\text{log}(\hat{y}_i)\\
&= -\frac{1}{N}\sum_{i=1}^N(y_i\text{log}(\frac{e^{g_i}}{\sum_{j=1}^Ne^{g_j}}))
\end{aligned}
\label{ce_soft}
\end{equation}

where $\hat{y}_i$ is the prediction of the recommendation model. Then, we calculate the Hessian Matrix of Eq.~\eqref{ce_soft}:
\begin{equation}
\begin{aligned}
if~\text{i}\neq\text{j},&~\mathcal{\bf M}_{ij}= -\hat{y}_i\hat{y}_j,\\
if~\text{i}=\text{j}, &~\mathcal{\bf M}_{ij}=\hat{y}_i-\hat{y}_i^2.
\nonumber
\end{aligned}
\label{eq:hessian}
\end{equation} 

Based on Eq.~\eqref{eq:hessian}, for any ${\bf c} \in \mathbb{R}^N$, we could have:
\begin{align}
\begin{aligned}
    {\bf c}^{\mathsf{T}}\mathcal{\bf M}{\bf c} &= \sum_{i=1}^N\hat{y}_i{\bf c}_i^2 - (\sum_{i=1}^N{{\bf c}_i}\hat{y}_i)^2 \nonumber\\ 
    &=  \sum_{i=1}^N\hat{y}_i{\bf c}_i^2 - (\sum_{i=1}^{N}{\bf c}_i\sqrt{\hat{y}_i}\sqrt{\hat{y}_i})^2
\nonumber    
\end{aligned}
\end{align}

Based on Cauchy-Schwarz Inequality~\cite{cauchy}, we scale the above equation:
\begin{align}
(\sum_{i=1}^N{\bf c}_i\sqrt{\hat{y}_i}\sqrt{\hat{y}_i})^2 \leq (\sum_{i=1}^N{{\bf c}_i}^2\hat{y}_i)(\sum_{i=1}^N\hat{y}_i).
\nonumber
\end{align}

Therefore, 
\begin{align}
    {\bf c}^{\mathsf{T}}\mathcal{\bf M}{\bf c} &\geq \sum_{i=1}^N\hat{y}_i{\bf c}_i^2 - (\sum_{i=1}^{N}{{\bf c}_i}^2\hat{y}_i)(\sum_{i=1}^N \hat{y}_i).
\nonumber
\end{align}

Since $\hat{y}_i$ is obtained by the softmax, $\sum_{i=1}^N\hat{y}_i = 1.$. Therefore:
\begin{align}
    {\bf c}^{\mathsf{T}}\mathcal{\bf M}{\bf c} &\geq 0.
\nonumber
\end{align}

Therefore, we have successfully demonstrated that the Hessian Matrix $\mathcal{\bf M}$ of Eq.~\eqref{ce_soft} is positive semi-definite (PSD). This proof is equivalent to establishing the convexity of $\mathcal{L}(y;{\bf g})$ with respect to $\bf g$.

For any ${\bf r} \in \mathbb{R}^N$ and $ \|{\bf r}\| = 1$, we could obtain:
\begin{equation}
\begin{aligned}
   {\bf r}^{\mathsf{T}}\mathcal{\bf M}{\bf r} &= \sum_{i=1}^N\hat{y}_i{\bf r}_i^2 - (\sum_{i=1}^N{{\bf r}_i}\hat{y}_i)^2 \nonumber\\  
   &=  \sum_{i=1}^N\hat{y}_i{\bf r}_i^2 - (\sum_{i=1}^{N}{\bf r}_i\sqrt{\hat{y}_i}\sqrt{\hat{y}_i})^2\\
    & \leq \sum_{i=1}^N{\bf r}_i^2 + (\sum_{i=1}^{N}{\bf r}_i\sqrt{\hat{y}_i}\sqrt{\hat{y}_i})^2. 
    \nonumber    
\end{aligned}
\end{equation}

Similarly, with Cauchy-Schwarz Inequality~\cite{cauchy}, we have: 
\begin{equation}
\begin{aligned}
    {\bf r}^{\mathsf{T}}\mathcal{\bf M}{\bf r} 
    & \leq \sum_{i=1}^N{\bf r}_i^2 + (\sum_{i=1}^{N}{\bf r}_i\sqrt{\hat{y}_i}\sqrt{\hat{y}_i})^2\\  &\leq  \sum_{i=1}^N{\bf r}_i^2 + (\sum_{i=1}^{N}{\bf r}_i^2)(\sum_{i=1}^{N}\hat{y}_i).
\nonumber    
\end{aligned}
\end{equation}

Since $ \|{\bf r}\| = 1$ and $\sum_{i=1}^N\hat{y}_i = 1$, we have 
\begin{align}
    {\bf r}^{\mathsf{T}}\mathcal{\bf M}{\bf r} 
    & \leq 2.
    \nonumber
\end{align}

Therefore, we have successfully proved that the eigenvalues of the Hessian Matrix $\mathcal{\bf M}$ are bounded by 2. This proof is equivalent to establishing the $\beta$-smoothness of $\mathcal{L}(y;{\bf g})$ with respect to $\bf g$.

After that, we calculate the first-order gradient of of $\mathcal{L}(y;{\bf g})$ over $\bf{g}$:
\begin{equation}
\begin{aligned}
&if~\text{i=c},\nabla_{\bf g}{\mathcal{L}(u, i, {y};{\bf g}})=-1 + \hat{y}_i,\\
&if~\text{i} \neq\text{c},\nabla_{\bf g}{\mathcal{L}(u, i, {y};{\bf g}})=\hat{y}_i.
\end{aligned}
\label{eq:first-order-gra} 
\end{equation}

We implement $\ell_2$ on Eq.~\eqref{eq:first-order-gra}:
\begin{equation}
\begin{aligned}
    \| \nabla_{\bf g}{\mathcal{L}({y};{\bf g}}) \| &= \sqrt{(-1+\hat{y}_c)^2+ \sum_{i\neq c}(\hat{y}_i)^2} \\ 
    &= \sqrt{\sum_{i\neq c}(\hat{y}_i)^2 + \sum_{i\neq c}(\hat{y}_i)^2} \\
    &\leq 2.
    \nonumber    
\end{aligned}
\end{equation}

Consequently, we have established that the $\ell_2$ norm of the first-order gradient of $\mathcal{L}(y;{\bf g})$ is bounded by a positive constant.

By demonstrating the convexity and $\beta$-smoothness of $\mathcal{L}(y;{\bf g})$ for all ${\bf g}$ and $y$, alongside the condition $|\nabla_{\bf g}\mathcal{L}(y;{\bf g}) | \leq B$ for all ${\bf g}$, where $B$ is a positive constant, we have successfully completed the proof of Theorem 1.

\end{proof}

\begin{theorem}
Assume that for all $(u,i,y)$, $\mathcal{L}_{SUP}(u,i,y;{\bm \theta})$ is differentiable, convex and $\beta$-smooth in ${\bm \theta}$, and both $\|\nabla_{{\bm \theta}}\mathcal{L}_{SUP}(u,i,y;{\bm  \theta}) \|,$\\
$ \|\nabla_{{\bm \theta}}\mathcal{L}_{SSL}(u,i;{\bm  \theta}) \| \leq B$ for all ${\bm \theta}$, where $B$ is a positive constant. With a fixed learning rate $\delta  = \frac{\epsilon}{\beta B^2}$, for every $(u,i,y)$ such that
\begin{equation}
\left \langle \nabla_{{\bm \theta}}\mathcal{L}_{SUP}(u,i,y;{\bm  \theta}),\nabla_{{\bm \theta}}\mathcal{L}_{SSL}(u,i;{\bm  \theta})  \right \rangle >\epsilon,
\nonumber 
\end{equation}
Thus,
\begin{equation}
\mathcal{L}_{SUP}(u,i,y;{\bm \theta}) > \mathcal{L}_{SUP}(u,i,y;{\bm \theta}(u,i)),
\nonumber 
\end{equation}
where ${\bm \theta}(u,i) = {\bm \theta} - \delta  \nabla_{{\bm \theta}}\mathcal{L}_{SSL}(u,i;{\bm  \theta})$.
\label{the_2}
\end{theorem}

\textit{Remark.} By applying \textit{Theorem}\ref{the_1}, we can assert that the objective loss function of the recommendation model possesses the properties of differentiability, convexity, and $\beta$-smoothness. Furthermore, in combination with \textit{Theorem}\ref{the_2}, we can conclude that test-time training could enhance the performance of the recommender system.

\begin{proof}
Following~\cite{sun2019test}, we prove Theorem 2 in this section. Specifically, for any $\delta $, by $\beta$-smoothness, we could obtain:
\begin{equation}
\begin{aligned}
\mathcal{L}_{SUP}(u,i,y;{\bm \theta(u,i)}) &= \mathcal{L}_{SUP}(u,i,y;{\bm \theta(x)} - \delta \nabla_{\bm \theta}{\mathcal{L}_{SSL}(u,i;{\bm \theta})}) \nonumber \\
    & \leq \mathcal{L}_{SUP}(u,i,y;{\bm \theta(u,i)}\\
    &- \delta \left \langle \nabla_{{\bm \theta}}\mathcal{L}_{SUP}(u,i, y;{\bm  \theta}),\nabla_{{\bm \theta}}\mathcal{L}_{SSL}(u,i;{\bm  \theta})  \right \rangle \nonumber\\ 
    &+ \frac{\delta ^2\beta}{2} \|\nabla_{\bm \theta}{\mathcal{L}_{SSL}(u,i;{\bm \theta})}\|^2.
    \nonumber  
\end{aligned}
\end{equation}

Then, we denote $\delta ^*$ as follows:
\begin{equation}
    \delta ^* = \frac{\left \langle \nabla_{{\bm \theta}}\mathcal{L}_{SUP}(u,i, y;{\bm  \theta}),\nabla_{{\bm \theta}}\mathcal{L}_{SSL}(u,i;{\bm  \theta})  \right \rangle }{\beta  \|\nabla_{\bm \theta}{\mathcal{L}_{SSL}(u,i;{\bm \theta})}\|^2},
    \nonumber
\end{equation}

Thus, we could calculate as:

\begin{flalign}
&\mathcal{L}_{SUP}(u,i,y;{\bm \theta} - \delta ^{*}\nabla_{\bm \theta}{\mathcal{L}_{SSL}(u,i;{\bm \theta})}) \nonumber \\ 
&\leq \mathcal{L}_{SUP}(u,i,y;{\bm \theta})  - \frac{{\left \langle \nabla_{{\bm \theta}}\mathcal{L}_{SUP}(u,i, y;{\bm  \theta}),\nabla_{{\bm \theta}}\mathcal{L}_{SSL}(u,i;{\bm  \theta})  \right \rangle }^2}{2\beta \|\nabla_{\bm \theta}{\mathcal{L}_{SSL}(u,i;{\bm \theta})}\|^2)}.
\nonumber
\end{flalign}

Based on the given assumptions regarding the $\ell_2$ norm of the first-order gradients for the main task and the self-supervised learning (SSL) task, as well as the assumption on their inner product, we can derive the following:

\begin{equation}
\begin{aligned}
 \mathcal{L}_{SUP}(u,i,y;{\bm \theta} ) &-\mathcal{L}_{SUP}(u,i,y;{\bm \theta} - \delta ^{*}\nabla_{\bm \theta}{\mathcal{L}_{SSL}(u,i;{\bm \theta})})\\ 
 &\geq \frac{\epsilon^2}{2\beta B^2}.
    \nonumber   
\end{aligned}
\end{equation}

Here, we denote $\delta  = \frac{\epsilon}{\beta B^2} $. Besides, based on the assumptions. we could obtain $0<\delta \leq\delta ^*$.

Then, denote ${\bf t}= \nabla_{\bm \theta}{\mathcal{L}_{SSL}(u,i;{\bm \theta})},$ by convexity of $\mathcal{L}_{SUP}$,
\begin{flalign}
&\mathcal{L}_{SUP}(u,i,y;{\bm \theta}(u,i)) = \mathcal{L}_{SUP}(u,i,y;{\bm \theta} - \delta {\bf t}) \nonumber \\
&= \mathcal{L}_{SUP}(u,i,y;(1-\frac{\delta }{\delta ^*}){\bm \theta}+ \frac{\delta }{\delta ^*}({\bm \theta}- \delta ^*{\bf t})) \nonumber \\
&\leq (1-\frac{\delta }{\delta ^*})\mathcal{L}_{SUP}(u,i,y;{\bm \theta}) + \frac{\delta }{\delta ^*}\mathcal{L}_{SUP}(u,i,y;{\bm \theta} - \delta ^*{\bf t}) \nonumber \\
&\leq (1-\frac{\delta }{\delta ^*})\mathcal{L}_{SUP}(u,i,y;{\bm \theta}) + \frac{\delta }{\delta ^*}(\mathcal{L}_{SUP}(u,i,y;{\bm \theta}) - \frac{\epsilon^2}{2\beta B^2})\nonumber \\
&= \mathcal{L}_{SUP}(u,i,y;{\bm \theta}) - \frac{\delta }{\delta ^*}\frac{\epsilon^2}{2\beta B^2}.
\nonumber
\end{flalign}

Moreover,  $\frac{\delta }{\delta ^*} > 0$, we could calculate: 
\begin{align}
    \mathcal{L}_{SUP}(u,i,y;{\bm \theta})-\mathcal{L}_{SUP}(u,i,y;{\bm \theta}(u,i)) > 0.
    \nonumber
\end{align}

Therefore, we can conclude that test-time training could enhance the performance of the recommender system.

\end{proof}

\begin{table}[]
\centering
\caption{Statistics of the datasets used in our method.}
\scalebox{0.95}{
\begin{tabular}{cccccc}
\hline
\textbf{Dataset}                            & \textbf{\#User}                & \textbf{\#Item}               & \textbf{\#IID int.}             & \textbf{\#OOD int.}             & \textbf{Density}              \\ \hline
\textbf{Synthetic data}                     & 1,000                         & 1,000                        & 145,270                        & 112,371                        & 0.257641                      \\
\textbf{Meituan}                            & 2,145                         & 7,189                        & 11,400                         & 6,944                          & 0.001189                      \\
\textbf{Yelp}                               & 7,975                         & 74,722                       & 305,128                        & 99,525                         & 0.000679                      \\
{\textbf{Amazon-Book}} & {11,000} & {9,332} & {75,138}  & {45,326}  & {0.001200} \\
{\textbf{Steam}}       & {23,310} & {5,237} & {135,961} & {180,229} & {0.002600} \\ \hline
\end{tabular}}
\label{dataset_info}
\end{table}

\section{Experiments}

We implement a series of experiments in this section to validate the effectiveness of the proposed DT3OR by addressing the following research questions:

\begin{itemize}
\item \textbf{RQ1}: How does the performance of DT3OR in OOD recommendation compared with other SOTA recommender methods?
\item \textbf{RQ2}: What is the impact of the hyper-parameters on the performance of DT3OR?
\item \textbf{RQ3}: How do the distinct components of DT3OR contribute to its performance?
\item \textbf{RQ4}: What is the underlying interest clustering structure revealed by DT3OR?
\item {\textbf{RQ5}: How does the efficiency of DT3OR compare to other recommendation algorithms?}
\item {\textbf{RQ6}: How stable is DT3OR under distribution shifts of different scales?}

\end{itemize}

\begin{table*}[]
\centering
\caption{{OOD recommendation performance comparison on three datasets with 11 baselines. Especially, the best results are emphasized in bold. Besides, $\dagger$ denotes results are statistically significant where the p-value is less than 0.05.}}
\scalebox{0.9}{
\begin{tabular}{cc|cccc|cccc|cccc}
\hline
\multicolumn{2}{c|}{\textbf{Dataset}}                                                & \multicolumn{4}{c|}{\textbf{Synthetic Data}}                                                                                              & \multicolumn{4}{c|}{\textbf{Meituan}}                                                                                                   & \multicolumn{4}{c}{\textbf{Yelp}}                                                                                       \\ \hline
\multicolumn{2}{c|}{\textbf{Metric}}                                                & \textbf{R@10}                     & \textbf{R@20}                & \textbf{N@10}                     & \textbf{N@20}                      & \textbf{R@50}                     & \textbf{R@100}              & \textbf{N@50}                     & \textbf{N@100}                    & \textbf{R@50}                & \textbf{R@100}               & \textbf{N@50}               & \textbf{N@100}              \\ \hline
\textbf{FM}                                   & ICDM 2010                            & 5.72                              & 10.74                        & 6.04                              & 7.92                               & 1.21                              & 2.05                        & 0.43                              & 0.57                              & 9.64                         & 13.89                        & 3.13                        & 3.85                        \\
\textbf{NFM}                                  & SIGIR 2017                           & 4.05                              & 7.61                         & 4.38                              & 5.60                               & 2.33                              & 3.54                        & 0.66                              & 0.85                              & 8.29                         & 12.76                        & 2.41                        & 3.16                        \\
\textbf{MultiVAE}                             & WWW 2018                             & 2.08                              & 4.08                         & 1.72                              & 2.57                               & 2.38                              & 3.68                        & 0.69                              & 0.91                              & 3.65                         & 5.82                         & 1.18                        & 1.54                        \\
\textbf{MacridVAE}                            & NeurIPS 2019                         & 2.31                              & 3.92                         & 1.92                              & 2.62                               & 2.19                              & 3.64                        & 0.67                              & 0.90                              & 4.08                         & 6.34                         & 1.35                        & 1.74                        \\
\textbf{MacridVAE+FM}                         & NeurIPS 2019                         & 4.63                              & 8.36                         & 5.13                              & 6.43                               & 2.33                              & 3.64                        & 0.66                              & 0.87                              & 4.07                         & 6.26                         & 1.40                        & 1.78                        \\
\textbf{CausPref}                             & WWW 2022                             & 5.41                              & 4.77                         & 6.74                              & 9.41                               & 3.44                              & 4.35                        & 0.98                              & 1.21                              & 13.46                        & 18.34                        & 4.14                        & 5.01                        \\
\textbf{COR}                                  & WWW 2022                             & 7.67                              & {14.43}                  & 8.04                              & 10.56                              & 3.68                              & {5.78}                  & 1.01                              & 1.35                              & {14.16}                  & {19.86}                  & {5.00}                  & {5.95}                  \\
{\textbf{InvCF}}         & {WWW 2023}      & {7.64}       & {13.09} & {7.42}       & {8.45}        & {2.98}       & {4.05} & {0.74}       & {0.98}       & {10.45} & {11.86} & {3.04} & {4.01} \\
{\textbf{PopGo}}         & {TOIS 2024}     & {7.08}       & {13.24} & {7.10}       & {9.52}        & {2.95}       & {4.28} & {0.85}       & {1.20}       & {12.68} & {14.62} & {3.57} & {4.16} \\
{\textbf{CausalDiffRec}} & {Preprint 2024} & {7.72}       & {14.00} & {{8.09}} & {{10.61}} & {3.14}       & {5.29} & {0.98}       & {1.16}       & {11.98} & {13.45} & {3.76} & {4.38} \\
{\textbf{DR-GNN}}        & {WWW 2024}      & {{7.95}} & {14.21} & {8.00}       & {10.34}       & {{3.72}} & {5.75} & {{1.04}} & {{1.35}} & {14.11} & {19.76} & {4.87} & {5.03} \\ \hline
\multicolumn{2}{c|}{Ours}                                                            & \textbf{8.20}$\dagger$                     & \textbf{14.84}$\dagger$               & \textbf{8.42}$\dagger$                     & \textbf{10.86}$\dagger$                     & \textbf{4.15}$\dagger$                     & \textbf{5.83}$\dagger$               & \textbf{1.13}$\dagger$                     & \textbf{1.40}$\dagger$                     & \textbf{14.41}$\dagger$               & \textbf{19.90}$\dagger$               & \textbf{5.10}$\dagger$               & \textbf{6.02}$\dagger$               \\
\multicolumn{2}{c|}{\%Improvement}                                                   & 3.14\%                            & 2.84\%                       & 4.08\%                            & 2.36\%                             & 11.56\%                           & 0.87\%                      & 8.65\%                           & 3.70\%                            & 1.77\%                       & 0.20\%                       & 2.00\%                      & 1.18\%                      \\ \hline
\end{tabular}}
\label{com_res}
\end{table*}

\begin{figure*}
\centering
\begin{minipage}{0.3\linewidth}
\centerline{\includegraphics[width=1\textwidth]{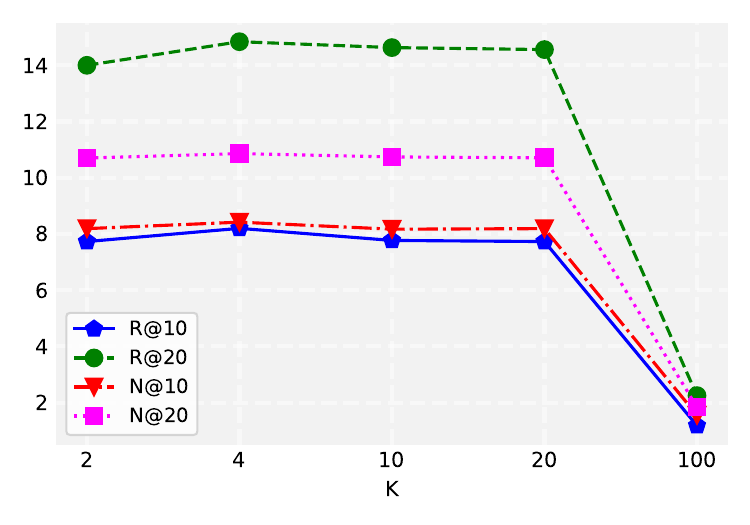}}
\vspace{3pt}
\centerline{{Synthetic}}
\end{minipage}\hspace{4mm}
\begin{minipage}{0.3\linewidth}
\centerline{\includegraphics[width=1\textwidth]{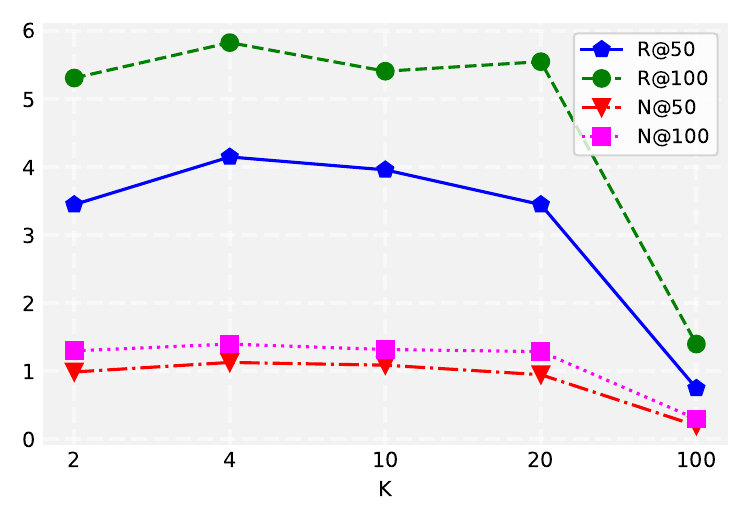}}
\vspace{3pt}
\centerline{{Meituan}}
\end{minipage}\hspace{4mm}
\begin{minipage}{0.3\linewidth}
\centerline{\includegraphics[width=1\textwidth]{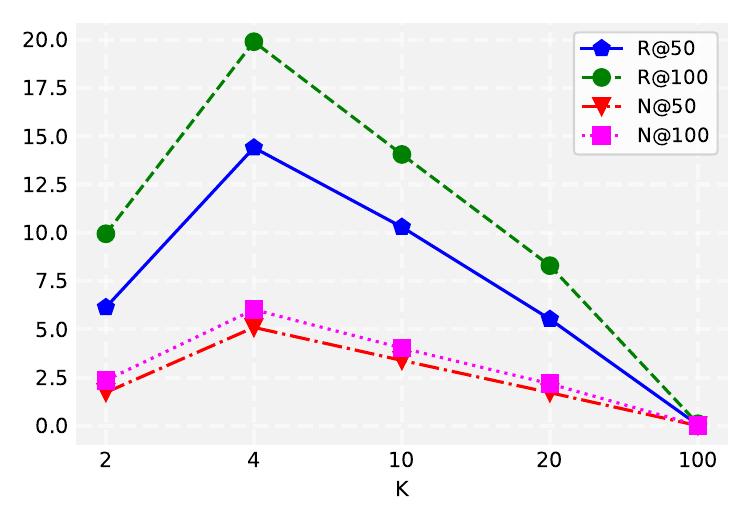}}
\vspace{3pt}
\centerline{{Yelp}}
\end{minipage}
\caption{Sensitivity analysis of the hyper-parameter cluster number $K$ on Synthetic, Meituan and Yelp datasets.}
\label{sen_K}
\end{figure*}

\subsection{Datasets \& Metric}

\subsubsection{Benchmark Datasets} {We conducted extensive experiments to evaluate the effectiveness of DT3OR using five benchmark datasets: Meituan\footnote{\url{https://www.biendata.xyz/competition/smp2021_2/}}, Yelp\footnote{\url{https://www.yelp.com/dataset.}}, Steam\footnote{\url{https://github.com/kang205/SASRec}}, Amazon-Book\footnote{\url{https://jmcauley.ucsd.edu/data/amazon/}}, and a synthetic dataset. To comprehensively assess the model, we constructed three distribution shift scenarios: synthetic shift, location shift, and temporal shift. The specific description is detailed as follows.} 

\begin{itemize}
    \item \textbf{Synthetic Shift.} The synthetic dataset contains 1,000 users and 1,000 items. The features of the users/items are sampled from the standard Gaussian $\mathcal{N}(0,1)$. The OOD features are collected by re-sampling user/item features from $\mathcal{N}(1,1)$. The detailed description of the dataset construction is provided in Appendix. \ref{Constrcution}.
    \item \textbf{Location Shift.} Yelp is a commonly utilized dataset for restaurant recommendations. Within this dataset, the user's location is taken into account as the fluctuating feature. The user ratings are sorted by timestamps, and the dataset is divided into two parts based on changes in location. The user interactions with shifted locations are considered OOD interactions.
    \item \textbf{Temporal Shift.} Meituan includes a public food recommendation dataset that contains a large number of user and item features, e.g., user consumption levels and food prices. For the OOD scenario in Meituan, we define it as a shift in the average consumption levels between weekdays and weekends. Specifically, users exhibit higher or lower consumption levels during weekends compared to weekdays. The user interactions during weekdays are considered in-distribution interactions, while those during weekends are considered OOD interactions.
    \item {\textbf{Temporal Shift.} The Amazon Book dataset offers a comprehensive collection of user ratings and reviews for books. It encompasses user IDs, item IDs, ratings, and in-depth reviews, along with extensive metadata including book titles, authors, categories, descriptions, and cover images. With 11,000 users, 9,332 items, and 120,464 interactions (45,326 OOD Int and 75, 138 IID Int). Similar to Meituan dataset, users exhibit higher consumption level during weekends compared with weekdays. Consequently, we classify weekday consumption as in-distribution and weekend consumption as out-of-distribution based on the time of consumption.}
    \item {\textbf{Temporal Shift.} The Steam dataset encompasses textual feedback from users on the Steam platform, featuring contributions from 23,310 users and covering 5,237 items, resulting in a total of 316,190 interactions (180,229 OOD Int and 135,961 IID Int). It includes user reviews, ratings, and diverse metadata related to games, such as titles, genres, and release dates. During weekdays, leisure time is limited, whereas weekends are better suited for relaxation. Similar to the case of Amazon Books, we also separate game time on Steam into two categories: weekday gaming time is considered in-distribution, while weekend gaming time is treated as out-of-distribution.}
\end{itemize}

For all datasets, interactions that have ratings of 4 or higher are treated as positive instances. To maintain a level playing field, a sparse partitioning approach is employed, allocating 80\%/10\%/10\% of the data for training, validation, and testing. Comprehensive statistics of the datasets used in this paper are listed in Tab.~\ref{dataset_info}.

\subsubsection{Evaluation Metrics} 
The performance of the recommendation is evaluated through two commonly applied metrics, i.e., Recall and NDCG. These metrics are evaluated using the all-ranking protocol~\cite{DGCF}, which considers the top-K items selected from the complete set of items that users have not interacted with. For the synthetic dataset, we designate K values of 10 and 20. Given the significant quantity of items in the Meituan and Yelp datasets, we showcase the results for K values of 50 and 100.

\subsection{Experiment Settings}
\subsubsection{Implementation Details}

The experiments are executed utilizing an NVIDIA 3090 GPU on the PyTorch platform. For the benchmark models, the performance values from the COR~\cite{cor} paper are directly reported for all datasets. In our proposed approach, we utilized COR as the primary network and trained our network using the Adam optimizer~\cite{adam}. The batch size was established at 500. Moreover, the training during test-time was carried out for 10 epochs. The trade-off hyper-parameter $\alpha$ was selected from the set $\{0.01, 0.1, 1.0, 10, 100\}$. The number of clusters $K$ was selected within the range of $\{2, 4, 10, 20, 100\}$, and the threshold $\tau$ for selecting high confidence samples was selected from $\{0.1, 0.3, 0.5, 0.7, 0.9\}$. The detailed hyper-parameter settings are provided in Tab.~\ref{hyper_para}.

\subsubsection{Comparison Recommendation Methods}
To substantiate the efficiency of our DT3OR, we contrast its performance with several baseline models, {including FM~\cite{FM}, NFM~\cite{NFM}, MultiVAE~\cite{MultiVAE}, MacridVAE~\cite{macidvae}, MacridVAE+FM, CausPref~\cite{causpref}, COR~\cite{cor}, InvCF~\cite{InvCF}, PopGo~\cite{PopGo}, CausalDiffRec~\cite{CausalDiffRec}, DR-GNN~\cite{DR-GNN}. }Specifically,

\begin{itemize}
    \item FM~\cite{FM}: With factorized parameters, FM models all interactions between variables.
    \item NFM~\cite{NFM}: NFM fuses the linearity of FM in handling second-order feature interactions with the non-linearity exhibited by networks.
    \item MultiVAE~\cite{MultiVAE}: MultiVAE used variational autoencoders to collaborative filtering for the implicit feedback.
    \item MacidVAE~\cite{macidvae}: It acquires disentangled representations by deducing high-level concepts linked to user intentions.
    \item CausPref~\cite{causpref}: CausPref used causal learning to obtain the invariant user preference and anti-preference negative sampling to handle the implicit feedback.
    \item COR~\cite{cor}: COR devised a new variational autoencoder to deduce the unobserved features from historical interactions in a causal way with out-of-distribution recommendation.
    \item {DR-GNN~\cite{DR-GNN} designed a distributionally robust optimization method to GNN-based recommendation algorithm.} 
    \item {PopGo~\cite{PopGo} proposed a debiasing strategy to reduce the interaction-wise popularity shortcut. It learned a shortcut model to yield a shortcut degree of a user-item pair based on representations.} 
    \item {CausalDiffRec~\cite{CausalDiffRec} addressed the out-of-distribution problem with causal diffusion method. By eliminating environmental confounding factors and learning invariant graph representations, the generalization ability of CausalDiffRec is improved for OOD data.}
    \item {InvCF~\cite{InvCF} designed an invariant collaborative filtering method to discover disentangled representations. It could reveal the latent preference and the popularity semantics.}
\end{itemize}


\begin{table*}[]
\centering
\caption{{OOD recommendation performance comparison on two datasets with 11 baselines. Especially, the best results are emphasized in bold. Besides, $\dagger$ denotes results are statistically significant where the p-value is less than 0.05.}}
\scalebox{1.1}{
\begin{tabular}{cc|cccc|cccc}
\hline
\multicolumn{2}{c|}{{\textbf{Dataset}}}                         & \multicolumn{4}{c|}{{\textbf{Amazon-Book}}}                                                                                          & \multicolumn{4}{c}{{\textbf{Steam}}}                                                                                                    \\ \hline
\multicolumn{2}{c|}{{\textbf{Method}}}                          & {\textbf{R@50}} & {\textbf{R@100}} & {\textbf{N@50}} & {\textbf{N@100}} & {\textbf{R@50}}  & {\textbf{R@100}} & {\textbf{N@50}} & {\textbf{N@100}} \\ \hline
{\textbf{FM}}            & {ICDM 2010}     & {2.19}          & {3.62}          & {0.47}          & {0.98}          & {9.52}           & {12.08}          & {2.47}          & {3.57}           \\
{\textbf{NFM}}           & {SIGIR 2017}    & {2.64}          & {3.89}          & {0.51}          & {1.03}          & {9.04}           & {11.85}          & {2.26}          & {2.68}           \\
{\textbf{MultiVAE}}      & {WWW 2018}      & {2.56}          & {3.57}          & {0.5}           & {1.11}          & {10.54}          & {12.62}          & {3.84}          & {4.76}           \\
{\textbf{MacridVAE}}     & {NeurIPS 2019}  & {2.98}          & {4.05}          & {0.67}          & {1.26}          & {11.47}          & {12.99}          & {4.04}          & {4.89}           \\
{\textbf{MacridVAE+FM}}  & {NeurIPS 2019}  & {3.04}          & {4.86}          & {1.01}          & {1.28}          & {12.75}          & {13.27}          & {5.04}          & {5.87}           \\
{\textbf{CausPref}}      & {WWW 2022}      & {3.84}          & {5.00}          & {1.12}          & {1.31}          & {12.84}          & {15.75}          & {5.17}          & {5.98}           \\
{\textbf{COR}}           & {WWW 2022}      & {4.12}          & {6.57}          & {1.24}          & {2.35}          & {13.44}          & {22.64}          & {5.72}          & {7.13}           \\
{\textbf{InvCF}}         & {WWW 2023}      & {4.32}          & {6.47}          & {1.28}          & {2.49}          & {{14.64}}    & {{23.16}}    & {{6.85}}    & {{8.52}}     \\
{\textbf{PopGo}}         & {TOIS 2024}     & {4.01}          & {6.25}          & {1.19}          & {2.23}          & {13.07}          & {20.46}          & {4.93}          & {6.10}           \\
{\textbf{CausalDiffRec}} & {Preprint 2024} & {{4.65}}    & {{6.98}}    & {{1.57}}    & {{2.63}}    & {14.27}          & {22.21}          & {5.97}          & {7.92}           \\
{\textbf{DR-GNN}}        & {WWW 2024}      & {4.23}          & {6.17}          & {1.35}          & {2.48}          & {13.15}          & {21.24}          & {5.06}          & {6.08}           \\ \hline
\multicolumn{2}{c|}{{\textbf{Ours}}}                            & {\textbf{5.07}$\dagger$} & {\textbf{7.56}$\dagger$} & {\textbf{2.26}$\dagger$} & {\textbf{2.85}$\dagger$} & {\textbf{15.72}$\dagger$} & {\textbf{24.17}$\dagger$} & {\textbf{7.65}$\dagger$} & {\textbf{9.73}$\dagger$}  \\
\multicolumn{2}{c|}{{\textbf{\%Improvement}}}                   & {9.03\%}        & {8.31\%}         & {4.08\%}        & {8.37\%}        & {7.38\%}         & {4.36\%}         & {11.68\%}       & {14.20\%}        \\ \hline
\end{tabular}}
\end{table*}

\subsection{(\textbf{RQ1})~Performance Comparisons}
We evaluate the experiments against 11 baseline models based on Recall and NDCG metrics. {These models fall into four classes, which include FM-based methods (FM~\cite{FM} and NFM~\cite{NFM}), VAE-based methods (MultiVAE~\cite{MultiVAE}, MacridVAE~\cite{macidvae}, and MacridVAE+FM), graph-based methods (DR-GNN~\cite{DR-GNN}), causal learning methods (InvCF~\cite{InvCF}, CausPref~\cite{causpref}, COR~\cite{cor}, CausalDiffRec~\cite{CausalDiffRec}, and PopGo~\cite{PopGo}).}

The comparison results are demonstrated in Tab.~\ref{com_res}. Observations drawn from the results are as follows:
\begin{itemize}
    \item Our proposed DT3OR could achieve SOTA recommendation performance compared to classical FM-based recommendation algorithms. We infer that the reason behind this is that these methods do not design specific strategy to handle the user/item feature shift; 
    \item FM-based recommendation algorithms have a significant advantage on datasets with notable feature shifts, resulting in better performance compared to VAE-based algorithms on synthetic and Yelp datasets;
    \item {Compared with graph-based and causal-based recommendation algorithms, our DT3OR demonstrates superior recommendation performance in out-of-distribution (OOD) scenarios. We attribute this improvement to the incorporation of test-time training with self-supervised learning (SSL) tasks, which enables our method to better adapt to shifted data distributions, leading to better performance.}
\end{itemize}

{In summary, our proposed DT3OR outperforms several recommendation methods across different scale datasets with different OOD scenario.} Taking the results on meituan dataset for example, DT3OR exceeds the second-best method by margins of 12.77\%, 0.87\%, 11.88\%, and 3.70\% in R@50, R@100, N@50, and N@100, respectively. The test-time-training with dual SSL tasks could assist the model adapt to the shift user/item features in test-time phase under OOD environment, thus alleviating the influence of data shift.

\subsection{(\textbf{RQ2})~Parameter Sensitivity}

We implement series of parameter analyses on Meituan, Yelp and the Synthetic datasets to verify the sensitivity of the hyper-parameters in this subsection.

\begin{figure*}
\centering
\begin{minipage}{0.3\linewidth}
\centerline{\includegraphics[width=1\textwidth]{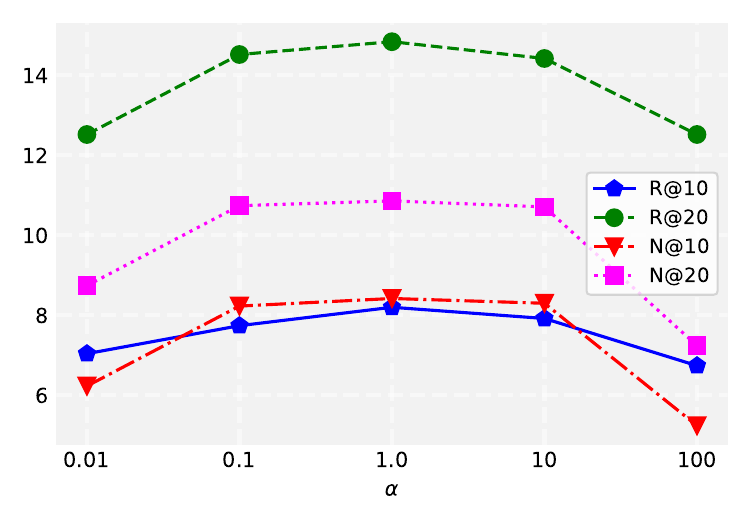}}
\vspace{3pt}
\centerline{{Synthetic}}
\end{minipage}\hspace{4mm}
\begin{minipage}{0.3\linewidth}
\centerline{\includegraphics[width=1\textwidth]{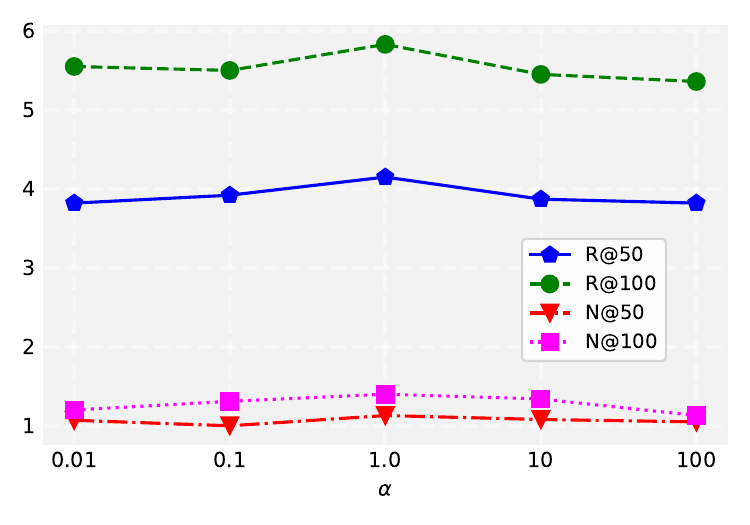}}
\vspace{3pt}
\centerline{{Meituan}}
\end{minipage}\hspace{4mm}
\begin{minipage}{0.3\linewidth}
\centerline{\includegraphics[width=1\textwidth]{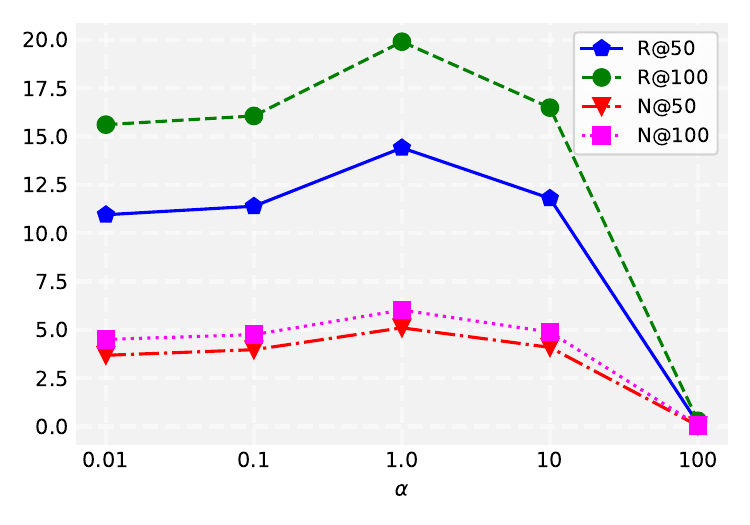}}
\vspace{3pt}
\centerline{{Yelp}}
\end{minipage}
\caption{Sensitivity analysis for the trade-off hyper-parameter $\alpha$.}
\label{sen_alpha}
\end{figure*}
\begin{figure*}
\centering
\begin{minipage}{0.3\linewidth}
\centerline{\includegraphics[width=1\textwidth]{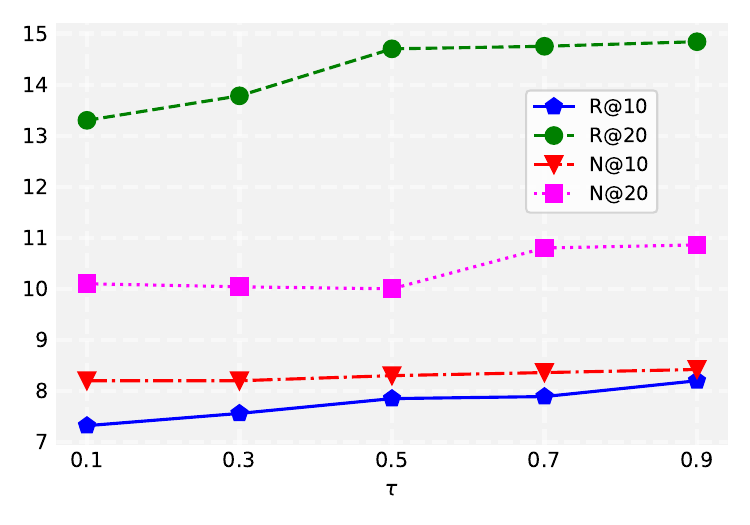}}
\vspace{3pt}
\centerline{{Synthetic}}
\end{minipage}\hspace{4mm}
\begin{minipage}{0.3\linewidth}
\centerline{\includegraphics[width=1\textwidth]{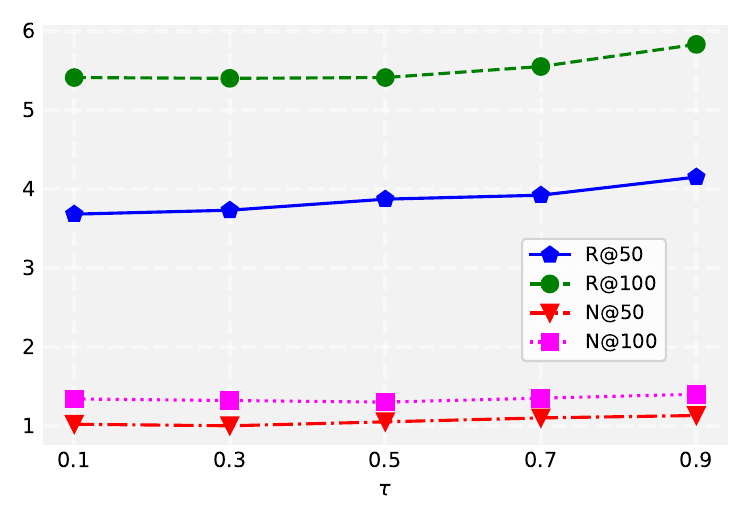}}
\vspace{3pt}
\centerline{{Meituan}}
\end{minipage}\hspace{4mm}
\begin{minipage}{0.3\linewidth}
\centerline{\includegraphics[width=1\textwidth]{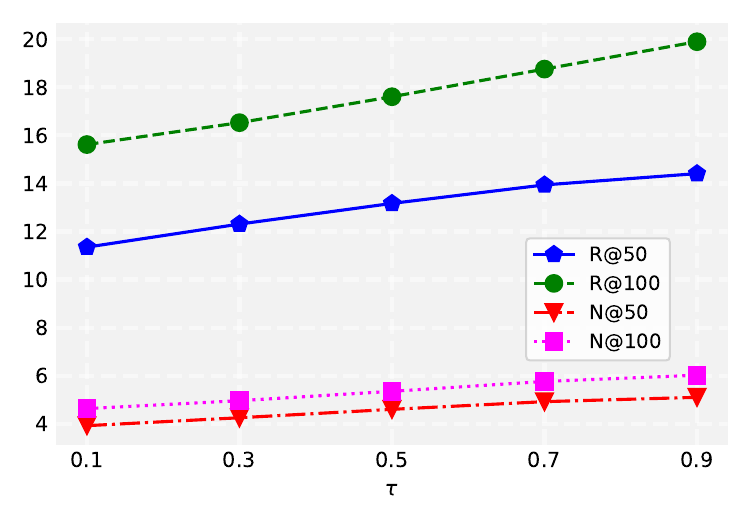}}
\vspace{3pt}
\centerline{{Yelp}}
\end{minipage}
\caption{Sensitivity analysis for the threshold hyper-parameter $\tau$.}
\label{sen_tau}
\end{figure*}
\begin{figure*}
\centering
\begin{minipage}{0.23\linewidth}
\centerline{\includegraphics[width=1\textwidth]{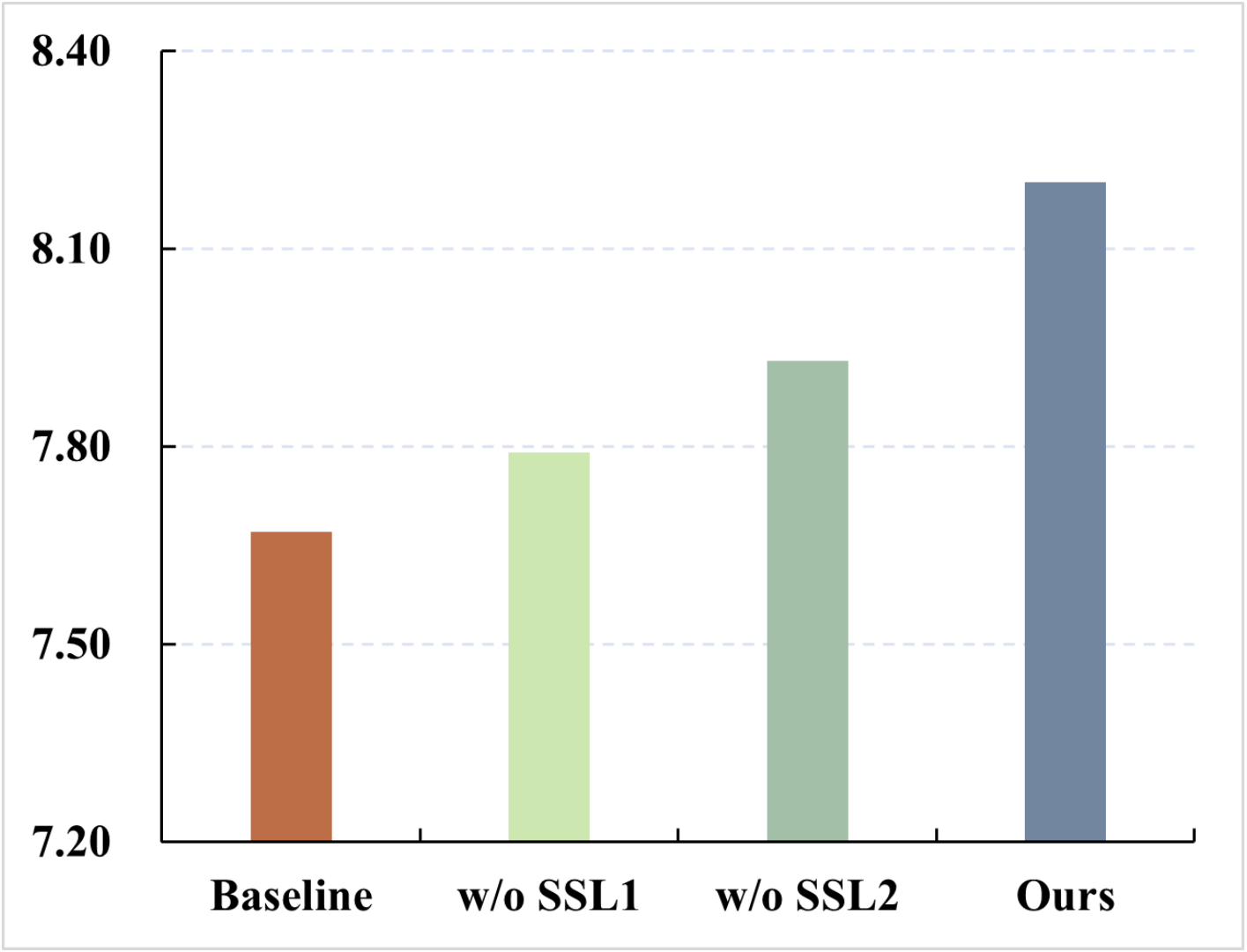}}
\centerline{{Synthetic with R@10}}
\vspace{10pt}
\centerline{\includegraphics[width=1\textwidth]{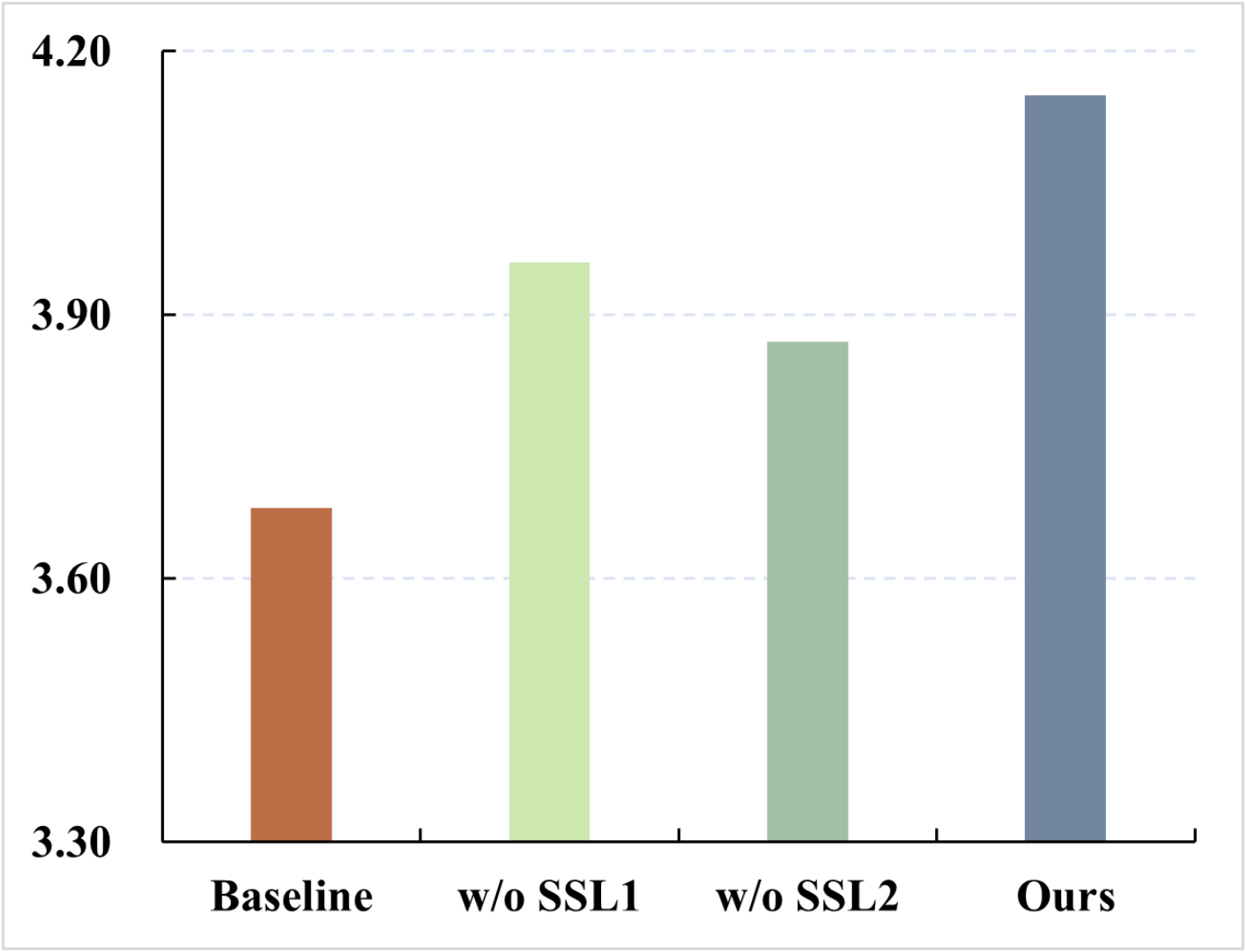}}
\centerline{{Meituan with R@50}}
\vspace{10pt}
\centerline{\includegraphics[width=1\textwidth]{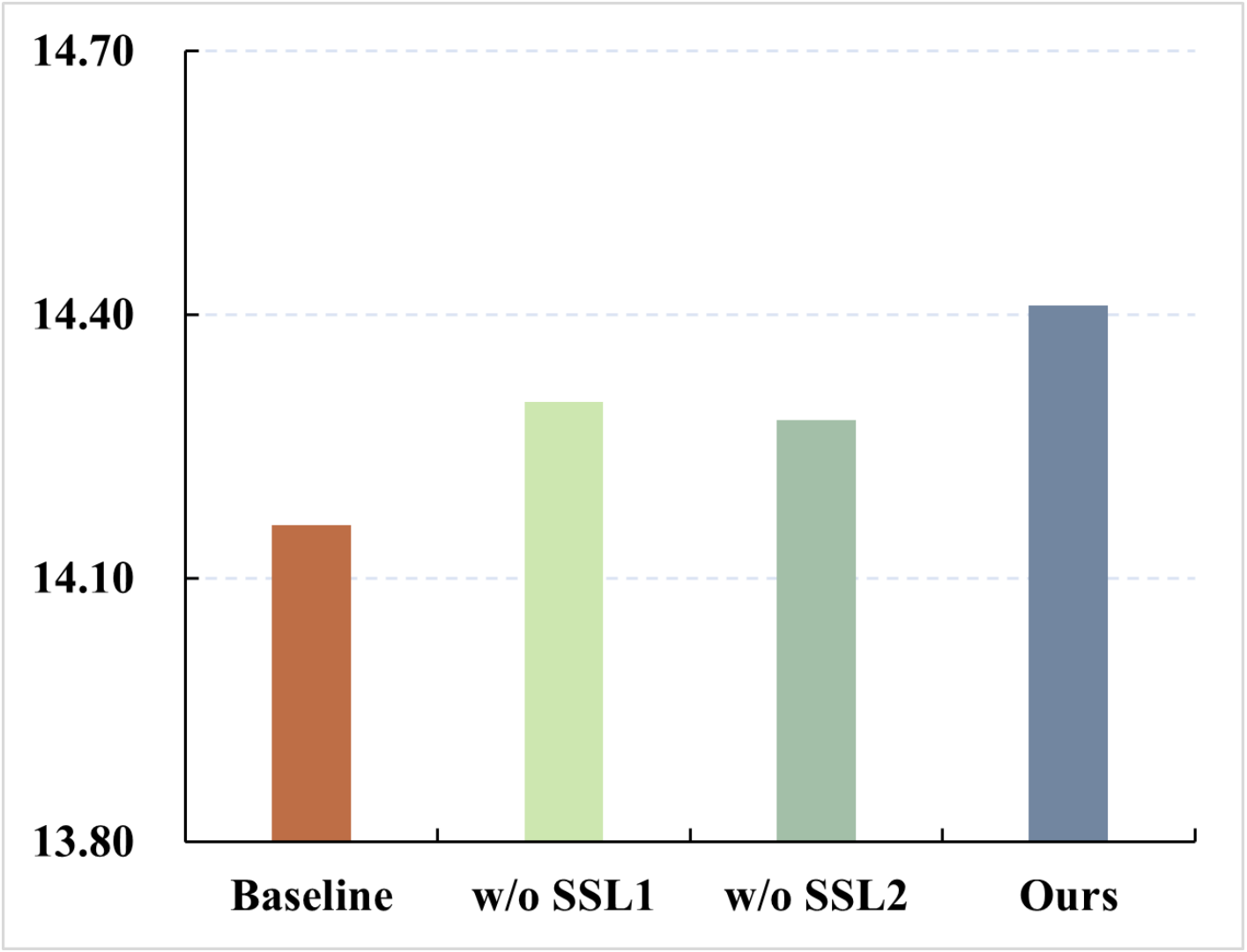}}
\centerline{{Yelp with R@50}}
\vspace{10pt}
\end{minipage}
\begin{minipage}{0.23\linewidth}
\centerline{\includegraphics[width=1\textwidth]{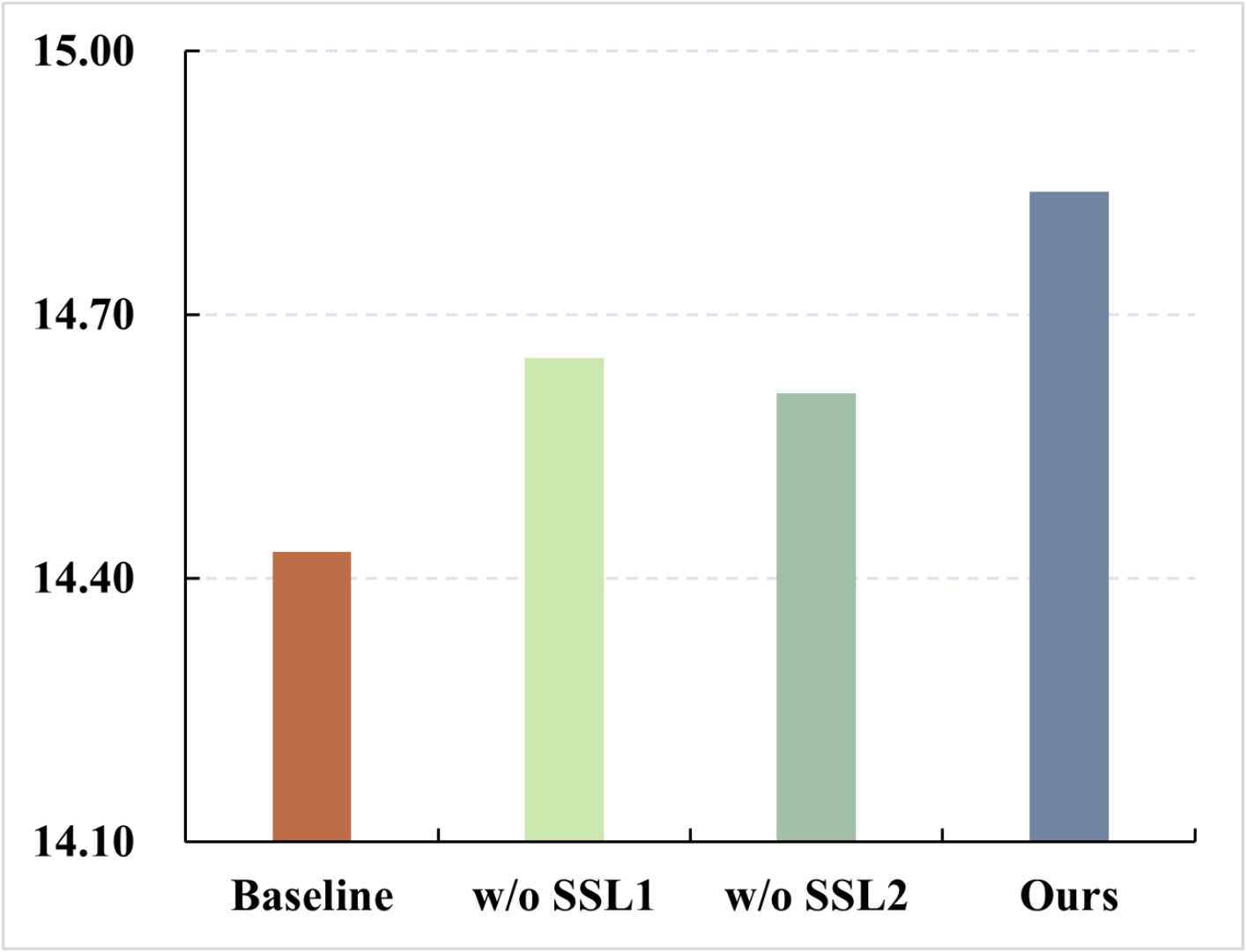}}
\centerline{{Synthetic with R@20}}
\vspace{10pt}
\centerline{\includegraphics[width=1\textwidth]{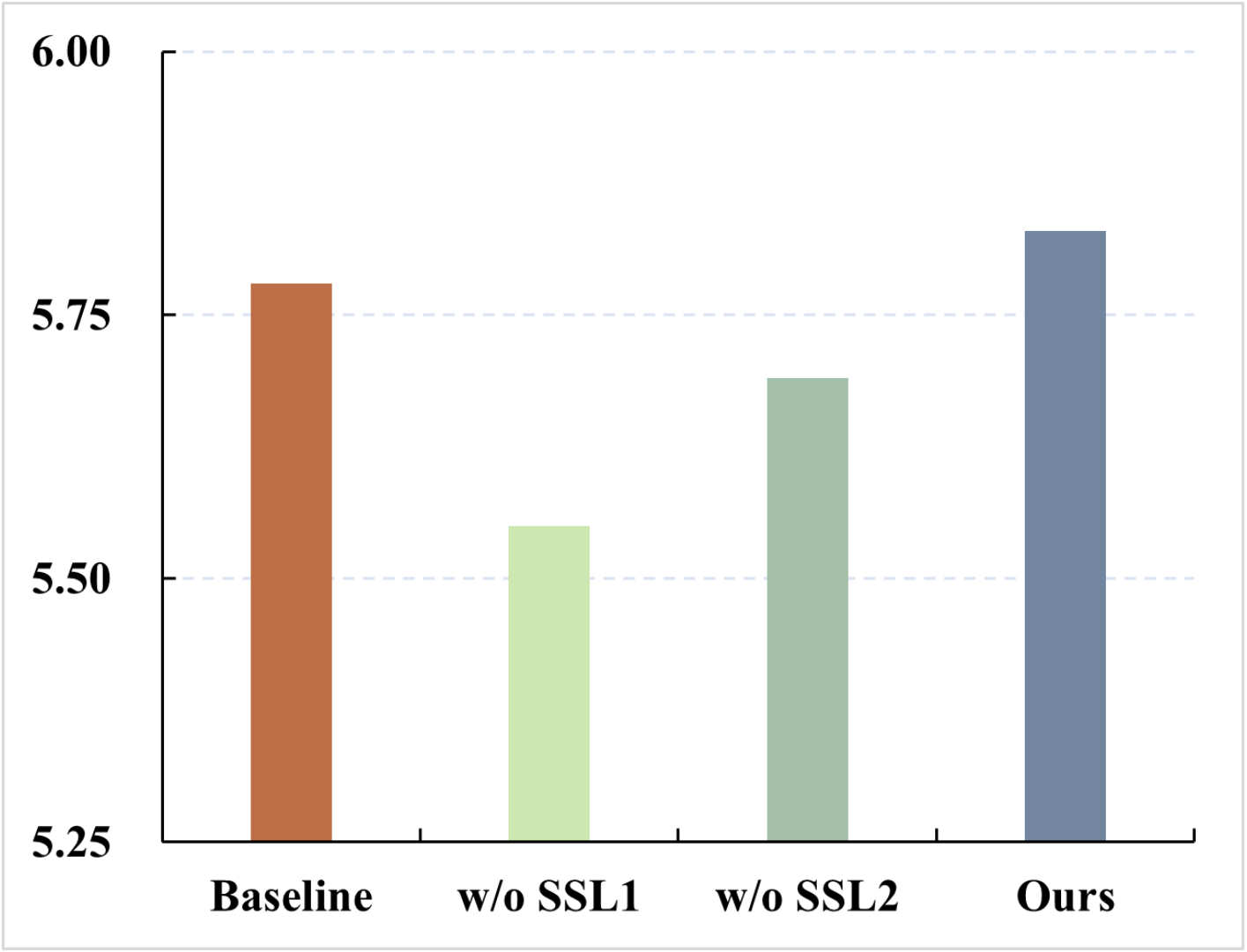}}
\centerline{{Meituan with R@100}}
\vspace{10pt}
\centerline{\includegraphics[width=1\textwidth]{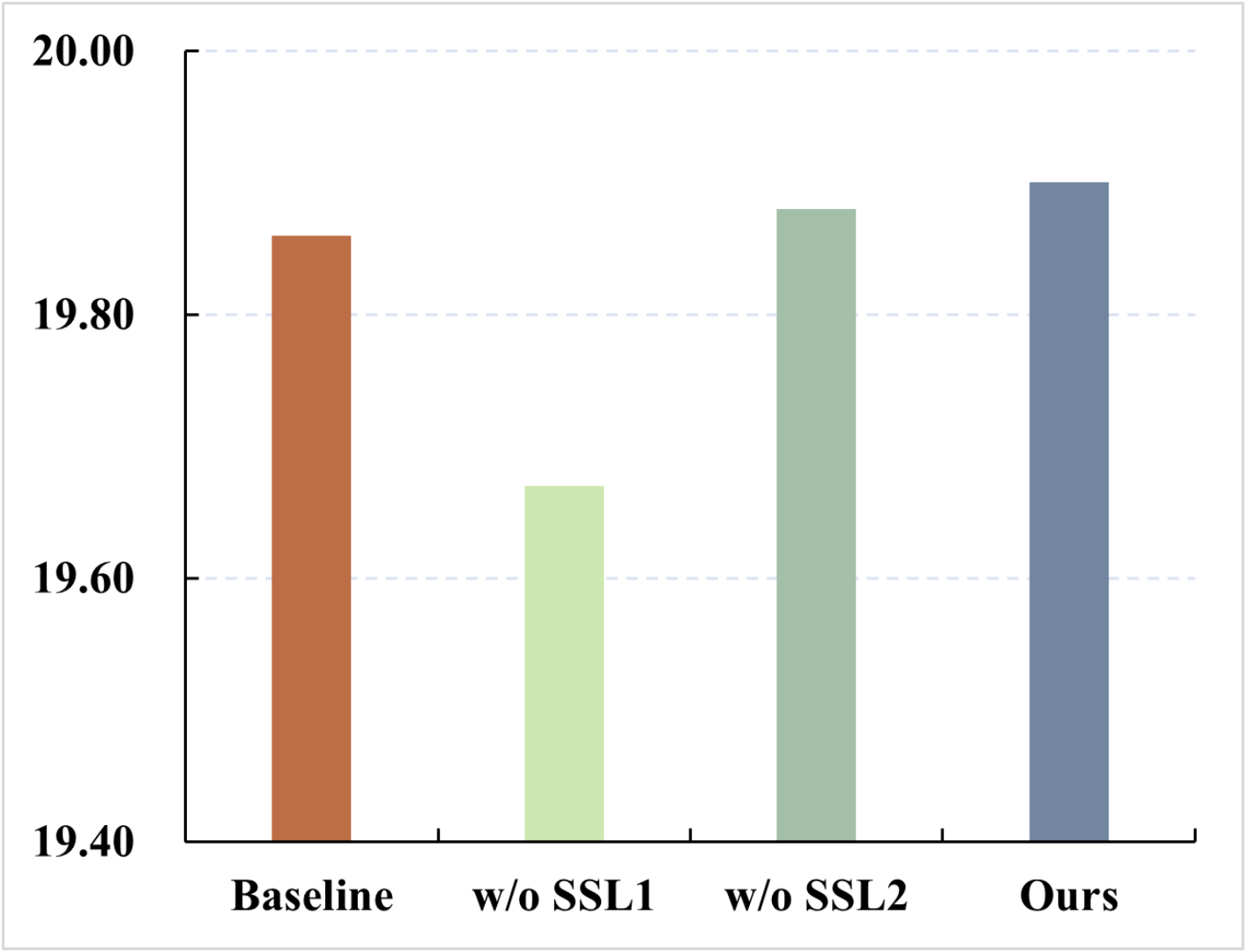}}
\centerline{{Yelp with R@100}}
\vspace{10pt}
\end{minipage}
\begin{minipage}{0.23\linewidth}
\centerline{\includegraphics[width=1\textwidth]{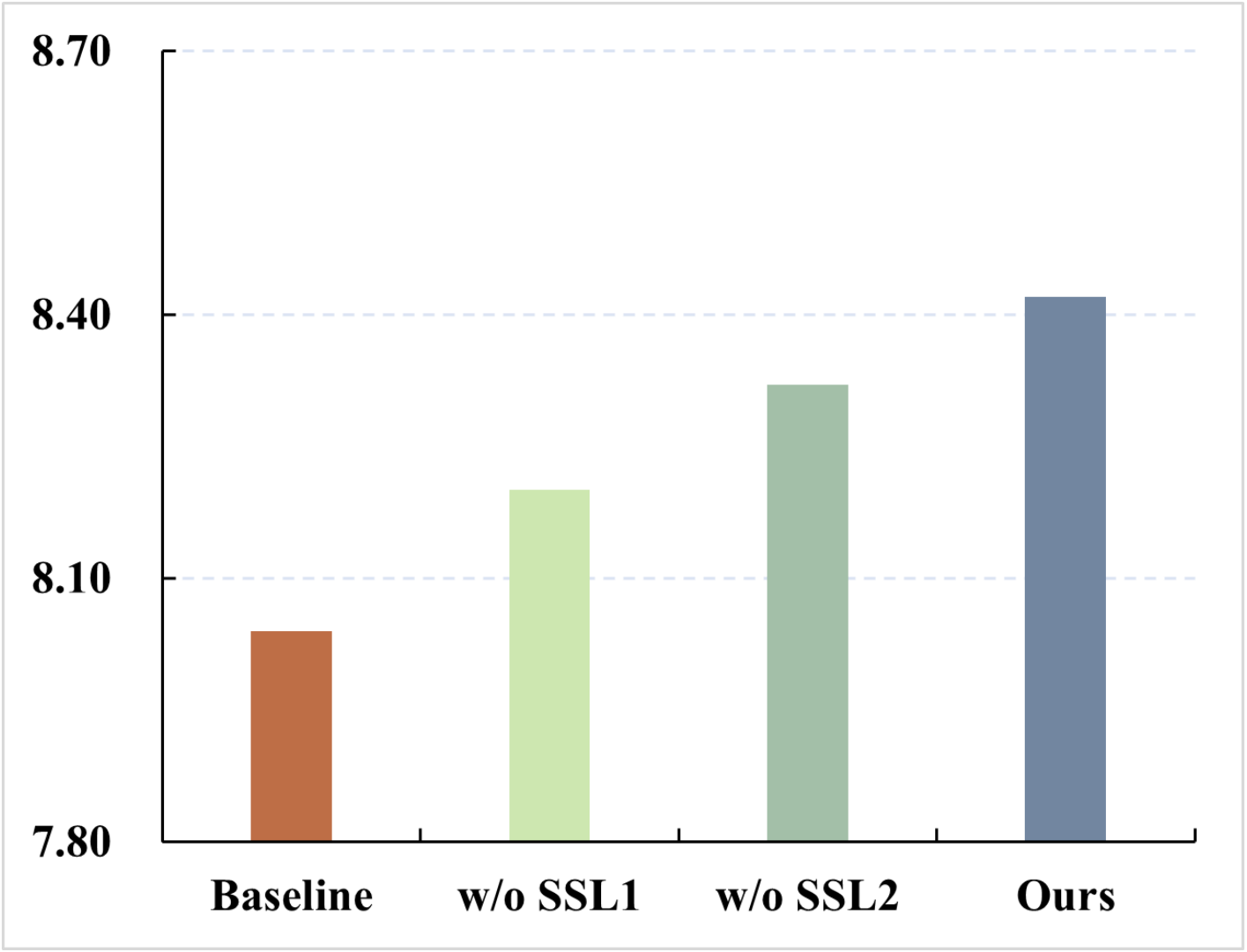}}
\centerline{{Synthetic with N@10}}
\vspace{10pt}
\centerline{\includegraphics[width=1\textwidth]{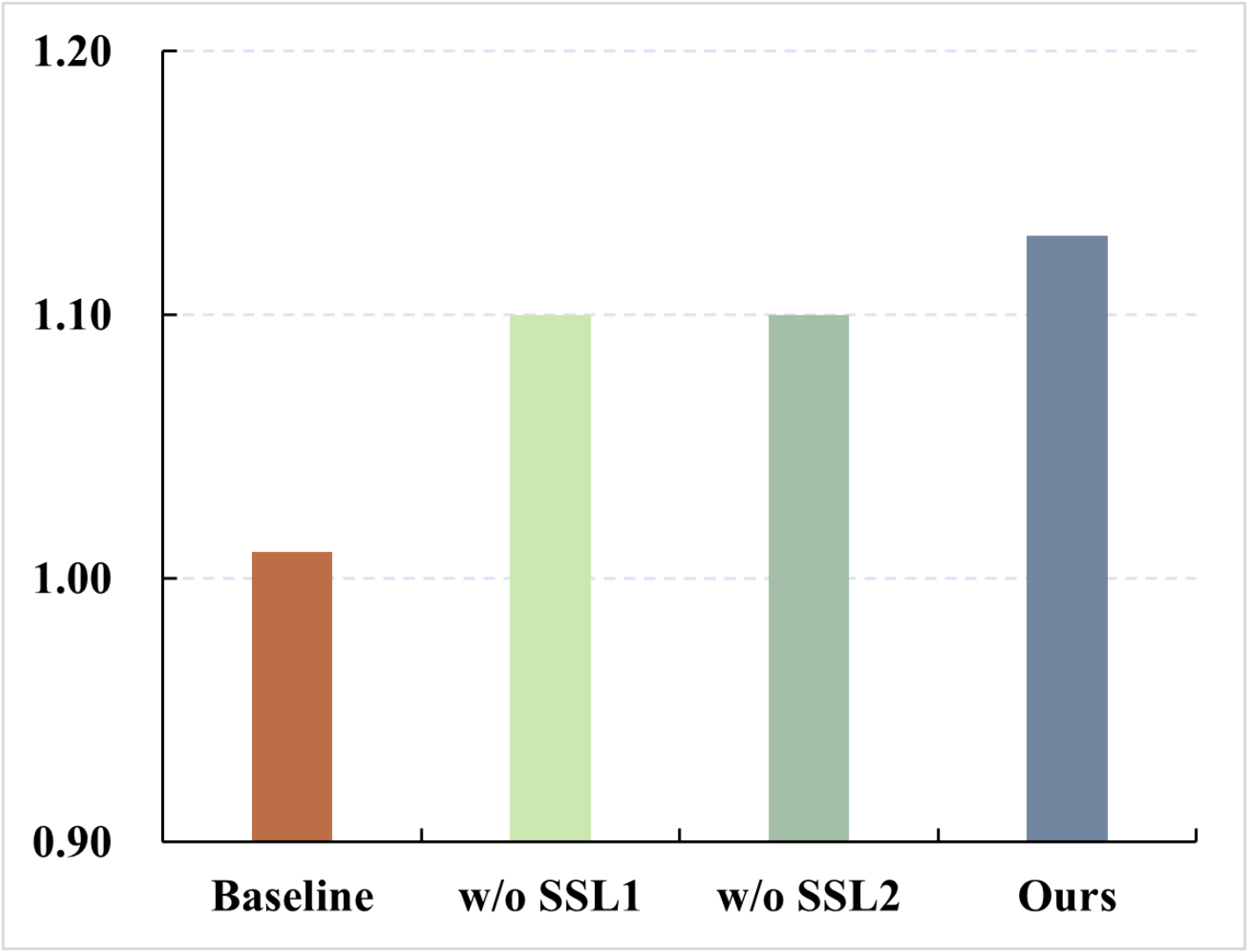}}
\centerline{{Meituan with N@50}}
\vspace{10pt}
\centerline{\includegraphics[width=1\textwidth]{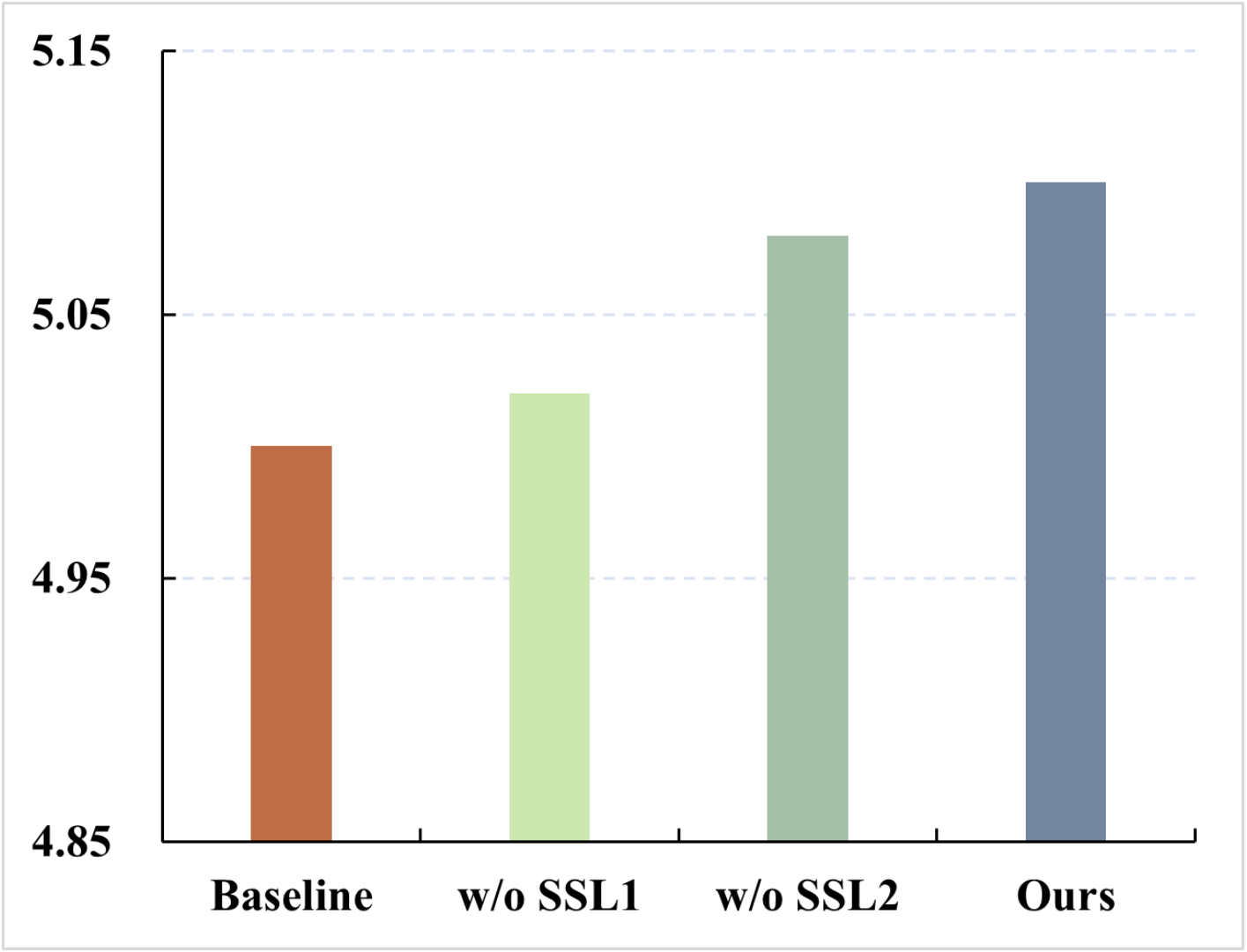}}
\centerline{{Yelp with N@50}}
\vspace{10pt}
\end{minipage}
\begin{minipage}{0.23\linewidth}
\centerline{\includegraphics[width=1\textwidth]{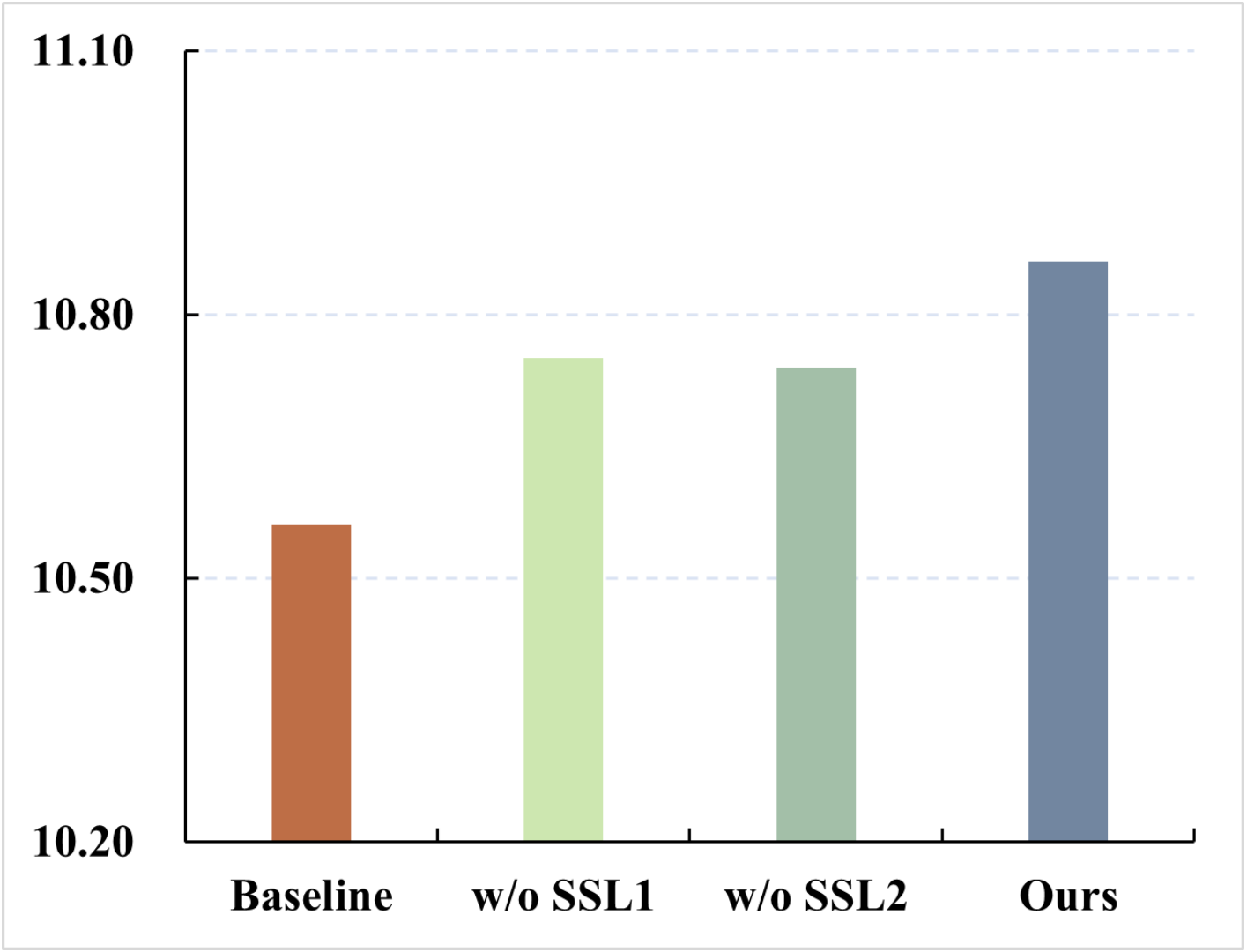}}
\centerline{{Synthetic with N@20}}
\vspace{10pt}
\centerline{\includegraphics[width=1\textwidth]{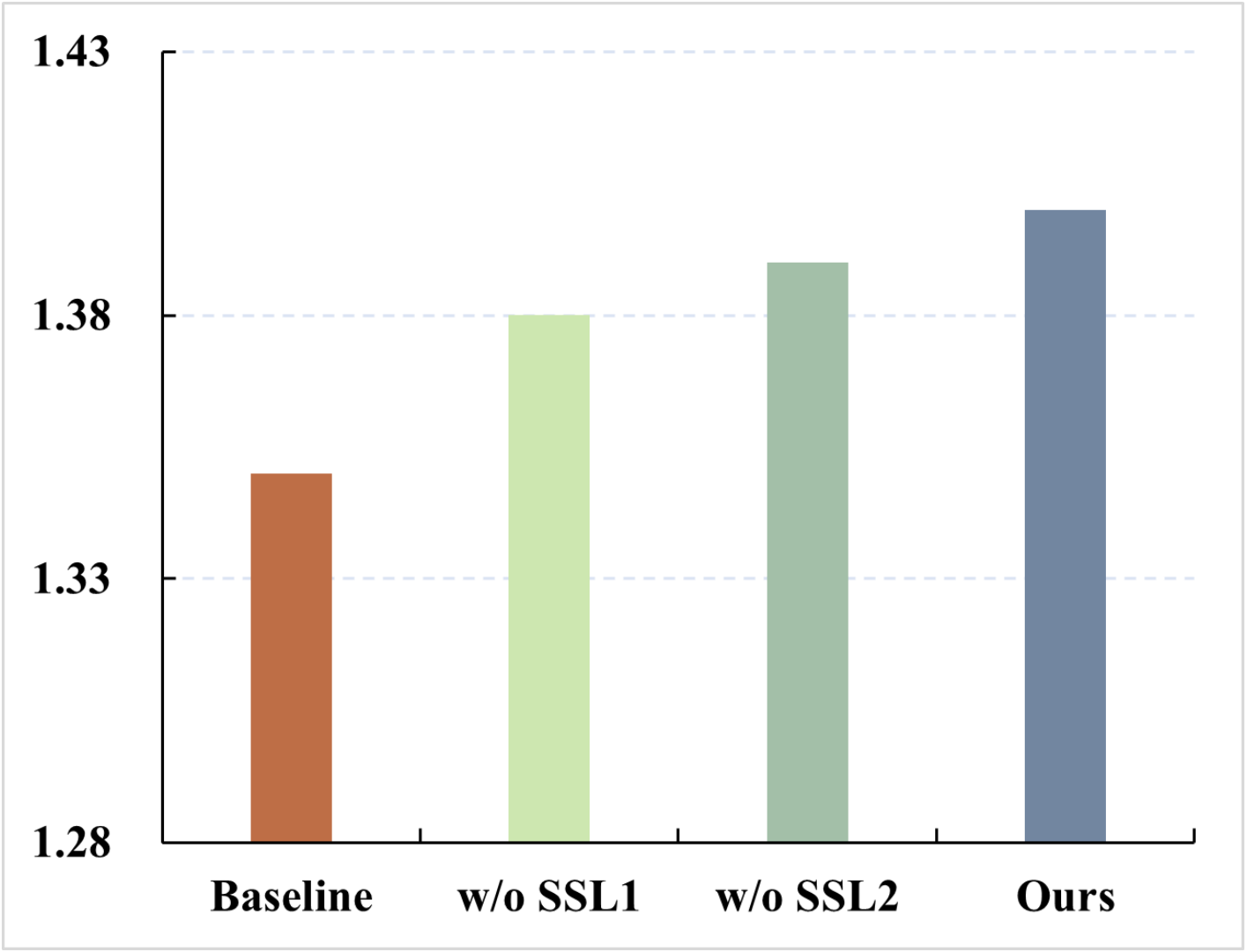}}
\centerline{{Meituan with N@100}}
\vspace{10pt}
\centerline{\includegraphics[width=1\textwidth]{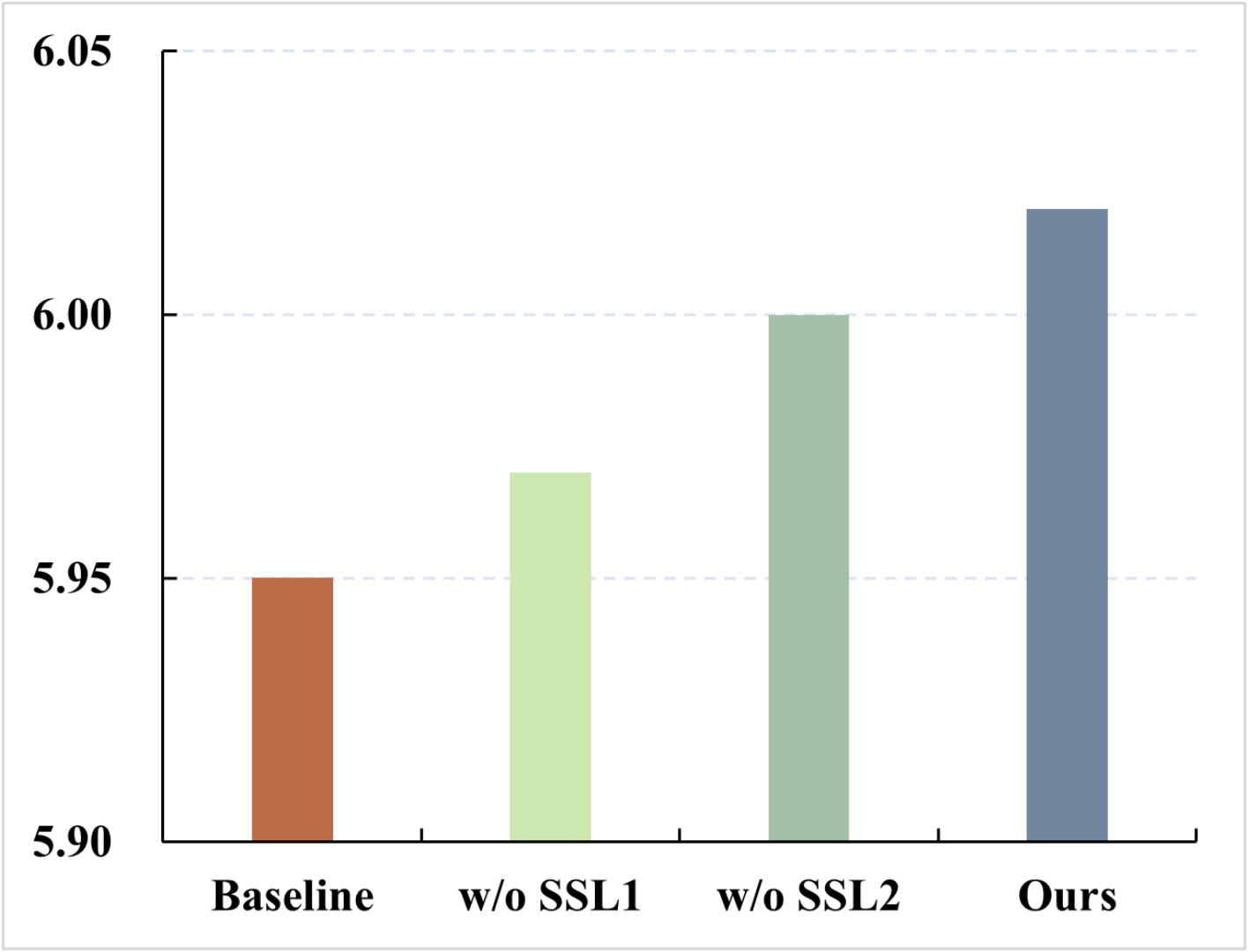}}
\centerline{{Yelp with N@100}}
\vspace{5pt}
\end{minipage}
\caption{Ablation studies over proposed DT3OR on three datasets. ``(w/o) SSL1'', ``(w/o) SSL2'' and ``(w/o) SSL1\&SSL2'' denote the reduced models by individually removing the self-distillation task, the contrastive task, and all aforementioned modules combined.}
\label{ablation_res}
\vspace{5pt}
\end{figure*}

\subsubsection{Sensitivity Analysis for trade-off $\alpha$}

To investigate the impact of the trade-off parameters $\alpha$ in our model, we performed experiments on the Synthetic, Meituan, and Yelp datasets. We varied the values of $\alpha$ within the range $\{0.01, 0.1, 1.0, 10, 100\}$. It is to be noted that all experimental settings remained consistent across different values of $\alpha$.
The results of these experiments are depicted in Fig.~\ref{sen_alpha}. The observations are as follows. 

\begin{itemize}
    \item When $\alpha$ is assigned extreme values, such as $0.01$ or $100$, the performance of the recommendation system tends to decrease. 
    \item Extreme values can disrupt the balance between different loss components, whereas values around $1.0$ tend to yield better results.
\end{itemize}

\begin{figure*}
\centering
\begin{minipage}{0.23\linewidth}
\centerline{\includegraphics[width=1\textwidth]{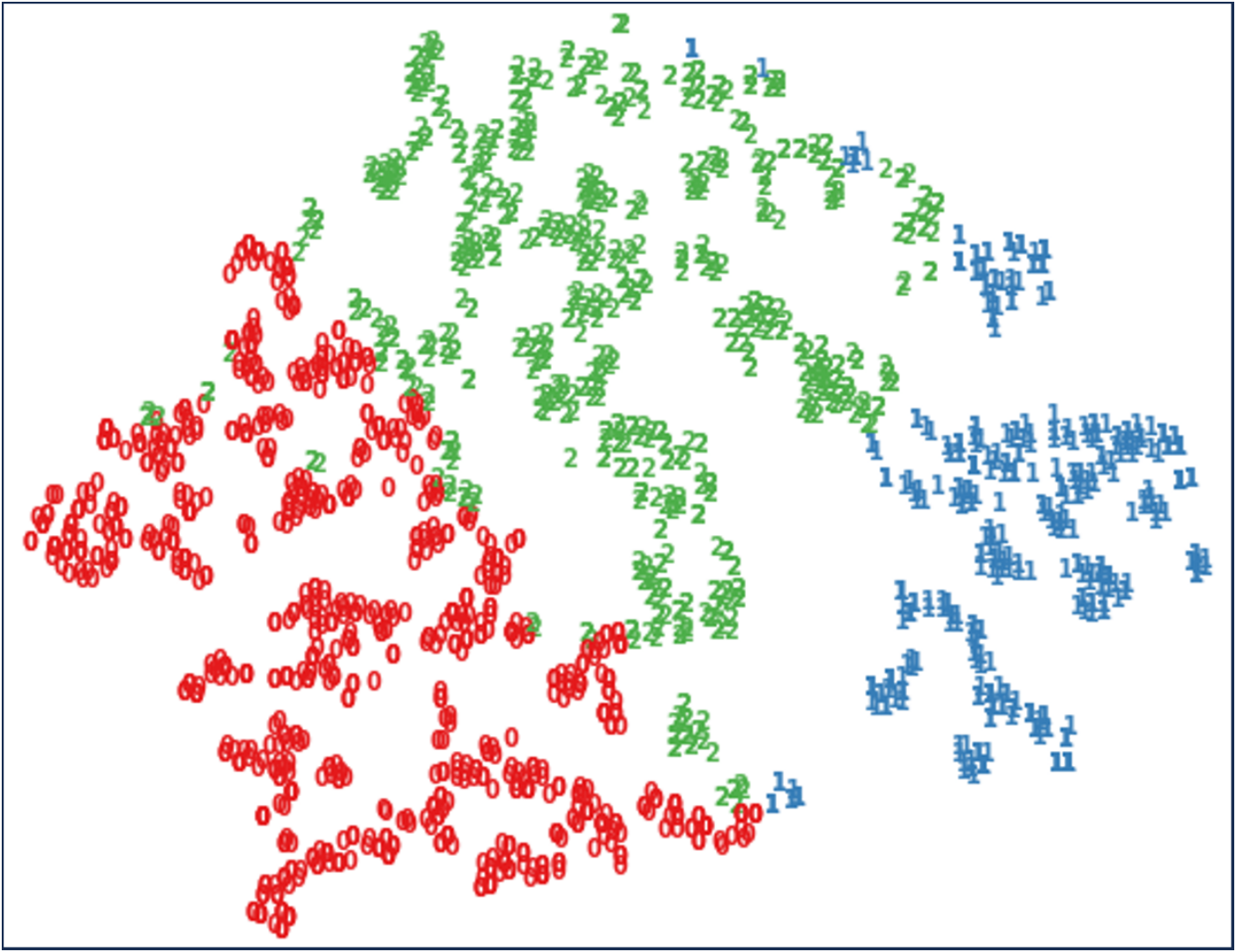}}
\vspace{5pt}
\centerline{{K = 3}}
\end{minipage}
\begin{minipage}{0.23\linewidth}
\centerline{\includegraphics[width=1\textwidth]{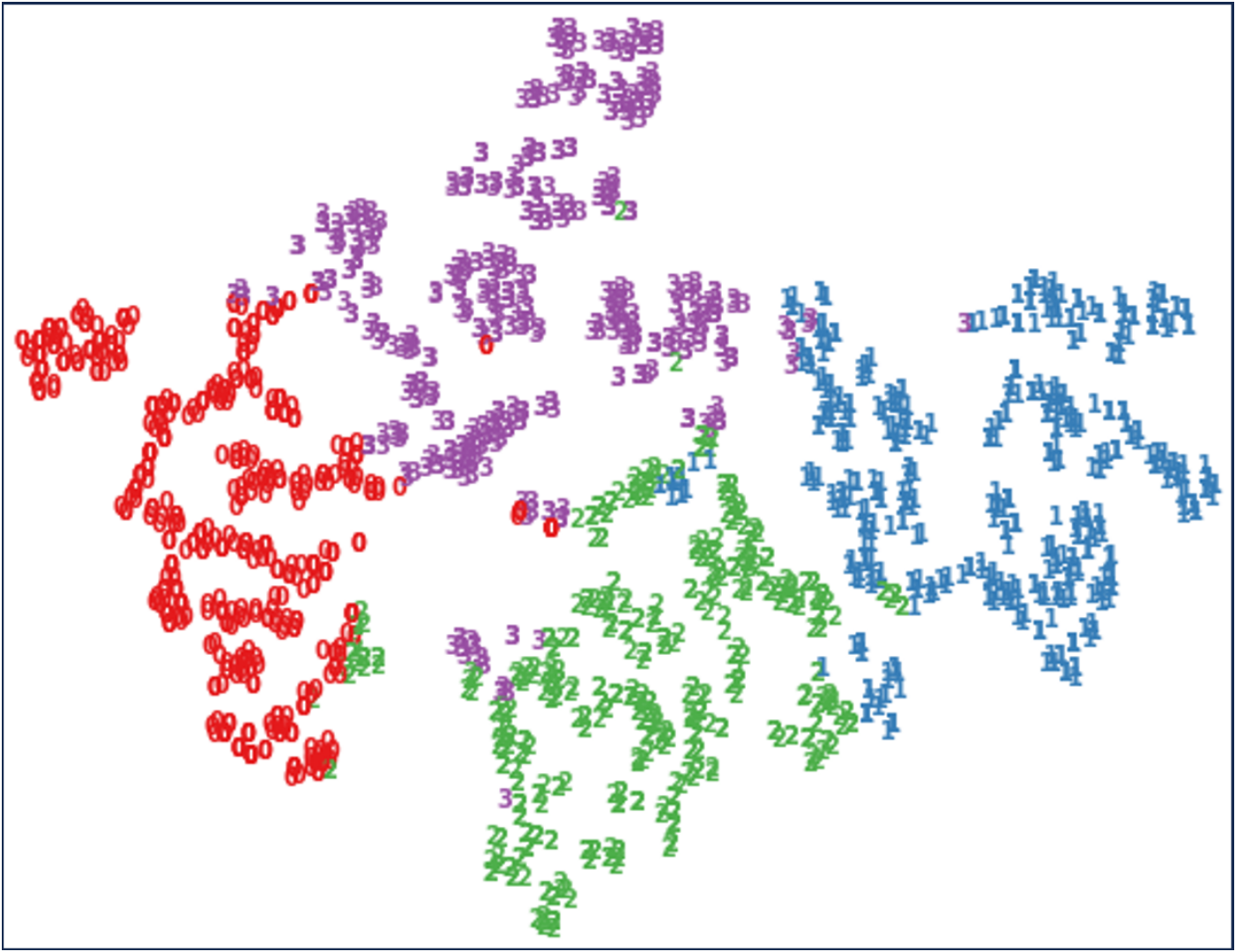}}
\vspace{5pt}
\centerline{{K = 4}}
\end{minipage}
\begin{minipage}{0.23\linewidth}
\centerline{\includegraphics[width=1\textwidth]{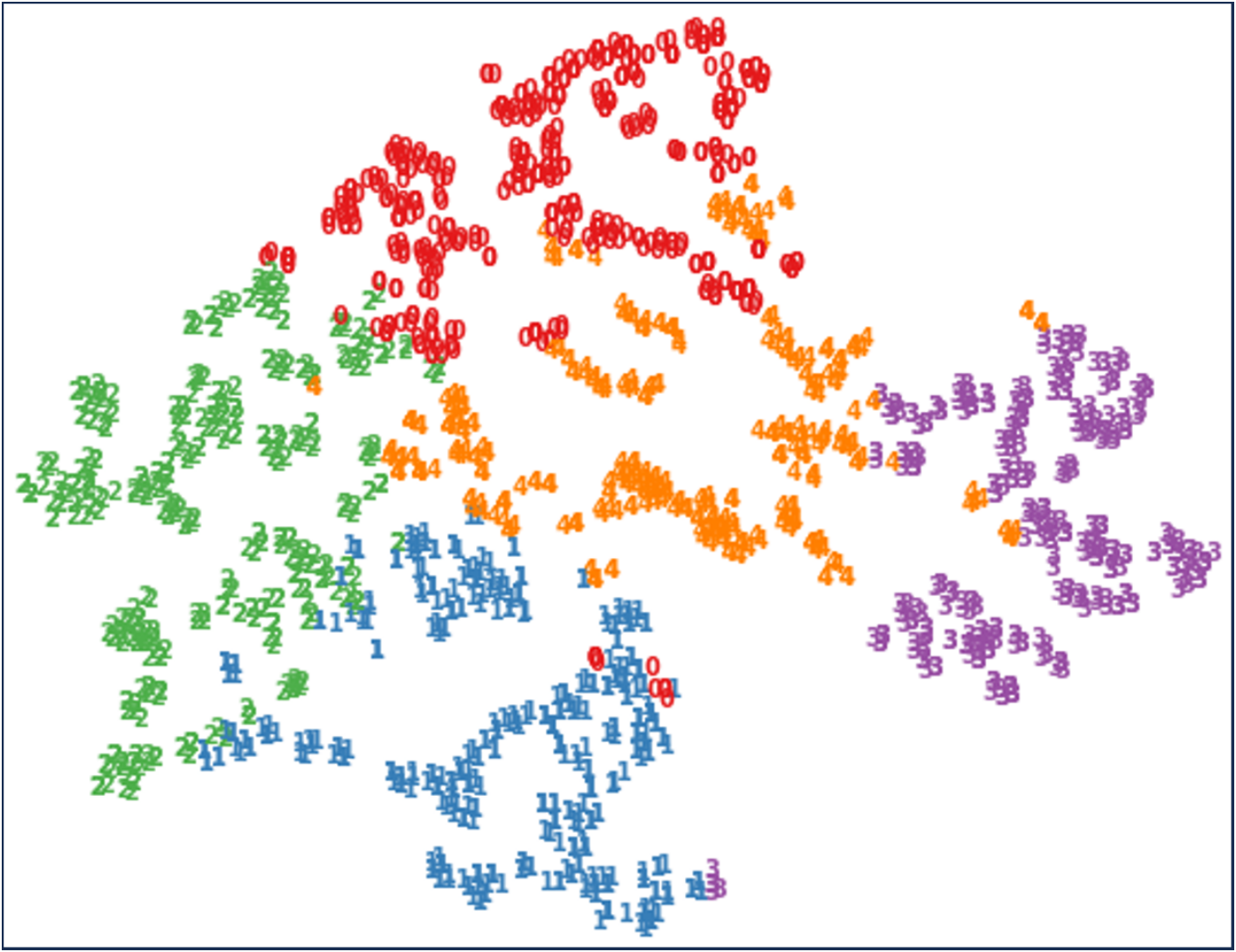}}
\vspace{5pt}
\centerline{{K = 5}}
\end{minipage}
\begin{minipage}{0.23\linewidth}
\centerline{\includegraphics[width=1\textwidth]{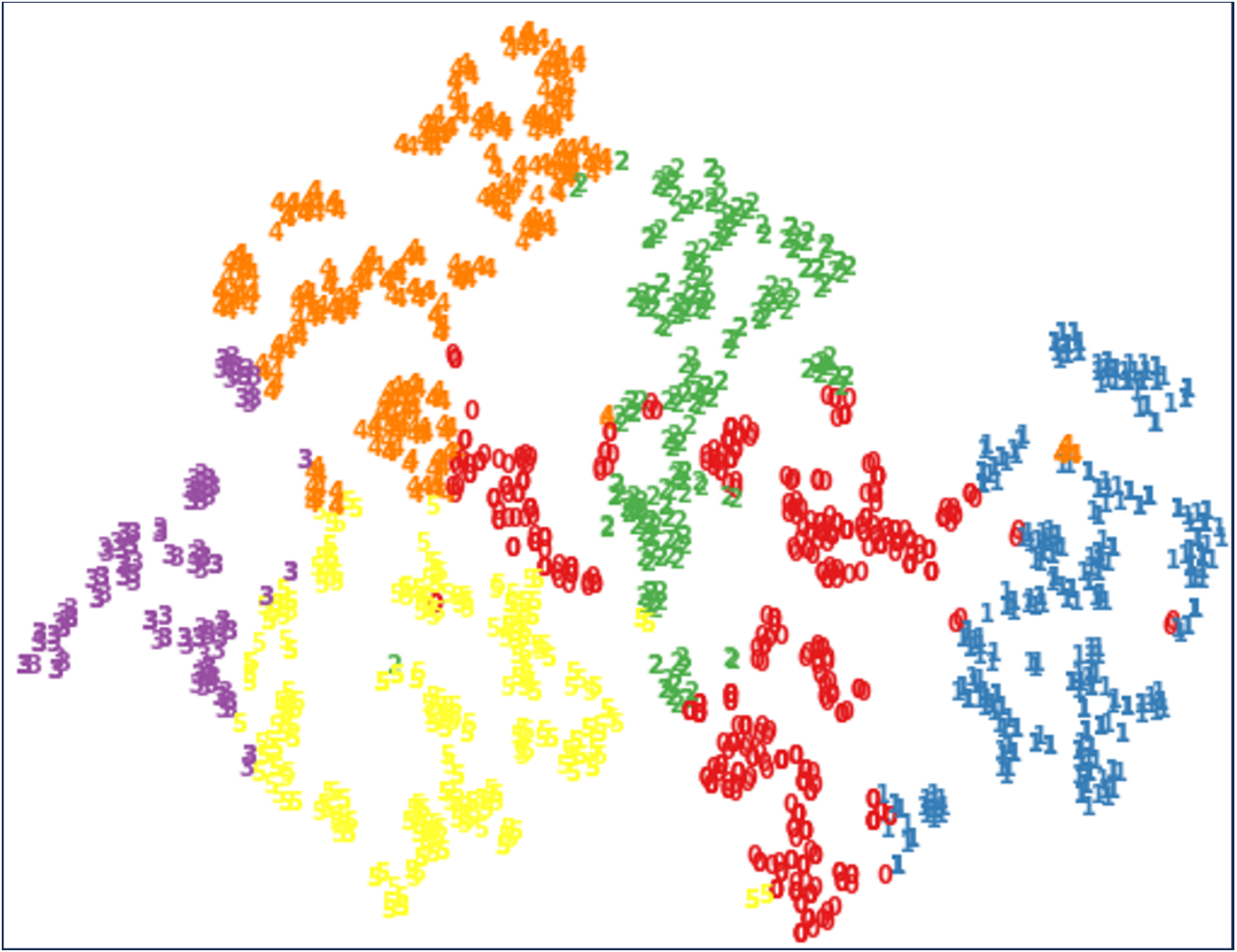}}
\vspace{5pt}
\centerline{{K = 6}}
\end{minipage}
\vspace{2pt}
\caption{Visualization experiment on Synthetic dataset with different cluster number K.}
\label{vis_res}
\end{figure*}

\subsubsection{Sensitivity Analysis for cluster number $K$}

We implement sensitive experiments to explore the impact of the parameter $K$, which represents the number of interest centers. We varied the value of $K$ within the range of $\{2, 4, 10, 20, 100\}$. According to the results in Fig.~\ref{sen_K}, we could obtain the following observations. 

\begin{itemize}
    \item The model achieves the best recommendation performance when $K$ is around 4. When $K$ takes extreme values, e.g., $K=100$, the performance will decrease dramatically. We speculate that this is because the interest centers become too scattered, making it difficult to accurately reflect the true preferences of users. 
    \item A similar situation occurs when $K=2$, where having too few interest centers fails to effectively capture the diverse preferences of users.
\end{itemize}

\subsubsection{Sensitivity Analysis for the threshold $\tau$}

We explore the impact of the threshold $\tau$ by implementing experiments on three datasets. The results are presented in Fig.~\ref{sen_tau}. We could conclude that the model achieves promising performance as $\tau$ increases. This can be attributed to the fact that the high-confidence centers $\textbf{C}_i$ become more reliable with a higher threshold.  

\begin{figure}[h]
\centering
\begin{minipage}{0.458\linewidth}
\centerline{\includegraphics[width=1\textwidth]{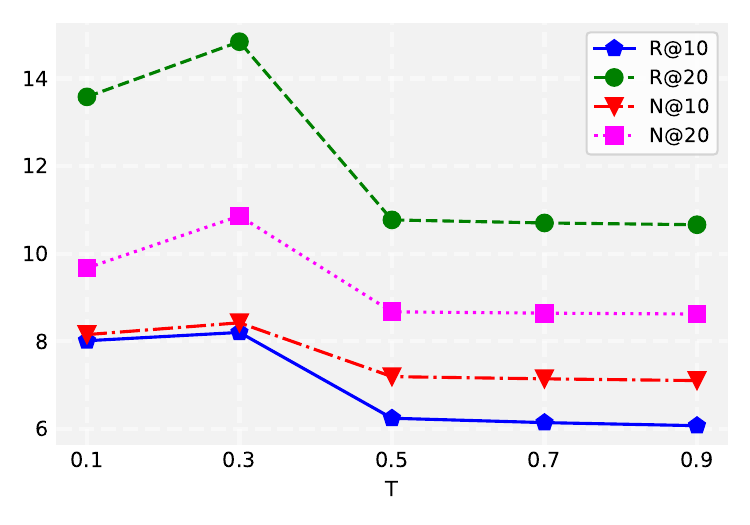}}
\centerline{{Synthetic}}
\vspace{10pt}
\end{minipage}
\begin{minipage}{0.48\linewidth}
\centerline{\includegraphics[width=1\textwidth]{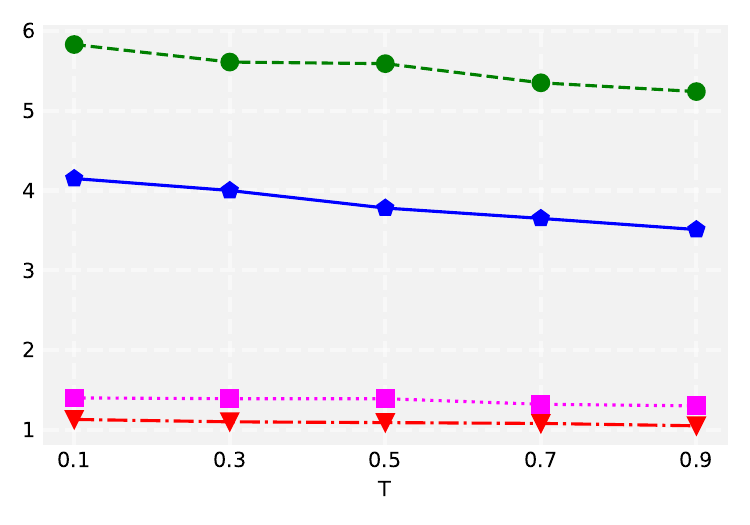}}
\centerline{{Meituan}}
\vspace{10pt}
\end{minipage}
\caption{{Sensitivity analysis of the hyper-parameter temperature $T$ on Synthetic and Meituan datasets.}}
\label{sen_T}
\end{figure}

\subsubsection{{Sensitivity Analysis for the temperature $T$}}
{We conduct the experiments to verify the sensitivity of the temperature parameter in sharpen function, i.e., Eq.~\eqref{sharpen_func}. To be specifically, we select the value of $T$ from the range of $\{0.1, 0.3, 0.5, 0.7, 0.9\}$. The results are presented in Fig.~\ref{sen_T}. 

We could observe that lowering the temperature could achieve better recommendation performance in out-of-distribution scenario. We analyze the reason is that when $T\to 0$, the output of the Sharpen function will approximate to a one-hot distribution. The entropy will decrease with lowering temperature. The specific setting of $T$ is shown in Tab.~\ref{hyper_para}.}

\subsection{(\textbf{RQ3})~Ablation Studies}
We implement ablation studies in OOD recommendation scenario to validate the effectiveness of the SSL tasks we designed: self-distillation task and contrastive task. Concretely, ``(w/o) SSL1'', ``(w/o) SSL2'' and ``(w/o) SSL1\&SSL2'' denote the reduced models where we individually removed the self-distillation task, the contrastive task, and all aforementioned modules combined. Here, we adopt COR~\cite{cor} as the backbone network to acquire representations for recommendation task, referred to as ``(w/o) Cau\&Con''. We have only removed the ablation module section while keeping the other modules unchanged from the DT3OR for conducting the ablation experiments. Specifically, apart from removing the corresponding modules, we keep identical parameter settings and training processes as used in the DT3OR. 

The results illustrated in Fig \ref{ablation_res} indicate that the removal of any of the designed modules results in a significant decrease in recommendation performance, which indicate the contribution of each module to the model's overall efficacy. We delve deeper into the reasons for these results as follows:
\begin{itemize}
    \item We utilize clustering task to obtain the interest centers of different types of users, which could better benefit the model to learn the preferences (user-like-item) in test-time phase of shift user/item features. 
    \item Benefiting from the data mining capacity, the contrastive mechanism could improve the discriminative capacity of the model, thus reducing the impact of the shift in user/item features.  
\end{itemize}

\subsection{(\textbf{RQ4})~Visualization Analysis}

To reveal the intrinsic interest clustering structure, we utilize visualization experiments to display the distribution of the learned embeddings in this subsection. To be concrete, experiments are implemented with $t$-SNE algorithm~\cite{T_SNE} on the Synthetic dataset. We perform visualization experiments with varying cluster numbers, i.e., $K \in [3,4,5,6]$. As presented in Fig.~\ref{vis_res}, the visual results emphasize that our proposed DT3OR method effectively identifies and represents the underlying interest clusters.

\subsection{{(\textbf{RQ5})~Time and Space Cost}}

{In this subsection, we conduct experiments to evaluate the time and space complexity using five datasets, i.e., Synthetic, Meituan, Yelp, Amazon, and Steam. The compared baselines are InvCF~\cite{InvCF}, COR~\cite{cor}, DR-GNN~\cite{DR-GNN}, and ours DT3OR. The experiments are conducted on the PyTorch deep learning platform with the 24G NVIDIA 3090. For the baselines, we utilized their source code available at Github. We measured time and space costs in seconds and gigabytes (GB), respectively. The results are reported by the average value over 10 runs. The observations are as follows.}

\begin{itemize}
\item {\textbf{Time Cost:} The experimental results for time cost are presented in Fig.\ref{time_space}. In this experiment, DT3OR utilizes COR\cite{cor} as the backbone network. We measured the training time for one epoch, averaged over 10 runs, with results reported in seconds. It can be observed that DT3OR's training time is comparable to that of other algorithms.}
\item {\textbf{Space Cost:} The space cost experimental results are shown in Fig.~\ref{time_space}. The results are reported as the average value over 10 runs and measured in gigabytes (GB). We observed that the memory costs of our DT3OR are acceptable when compared to other algorithms. Furthermore, these findings validate the effectiveness of the theoretical complexity analysis discussed in Sub-Section.~\ref{complex_loss}.}
\end{itemize}

\begin{figure}
\centering
\begin{minipage}{0.45\linewidth}
\centerline{\includegraphics[width=1\textwidth]{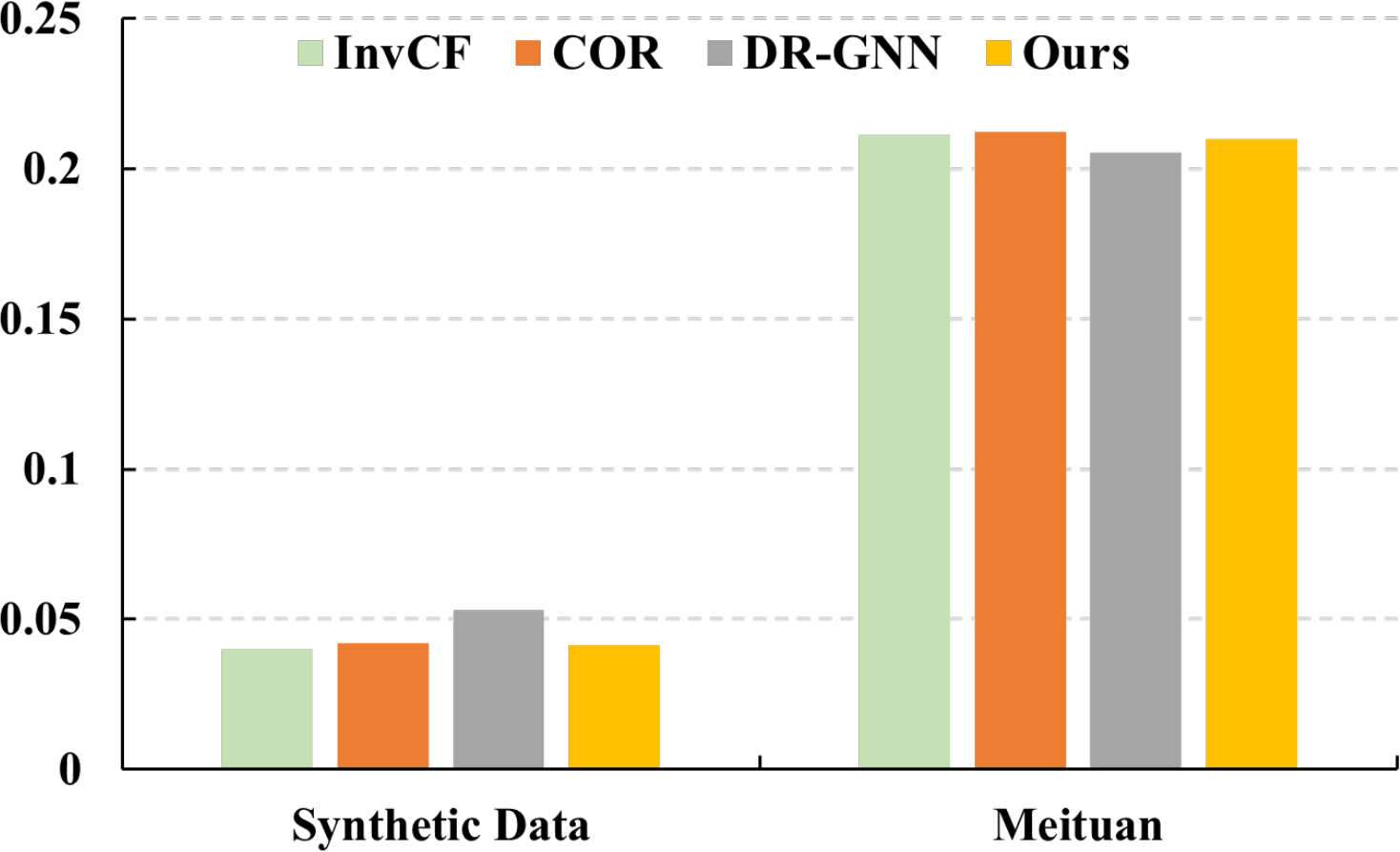}}
\vspace{10pt}
\centerline{\includegraphics[width=1\textwidth]{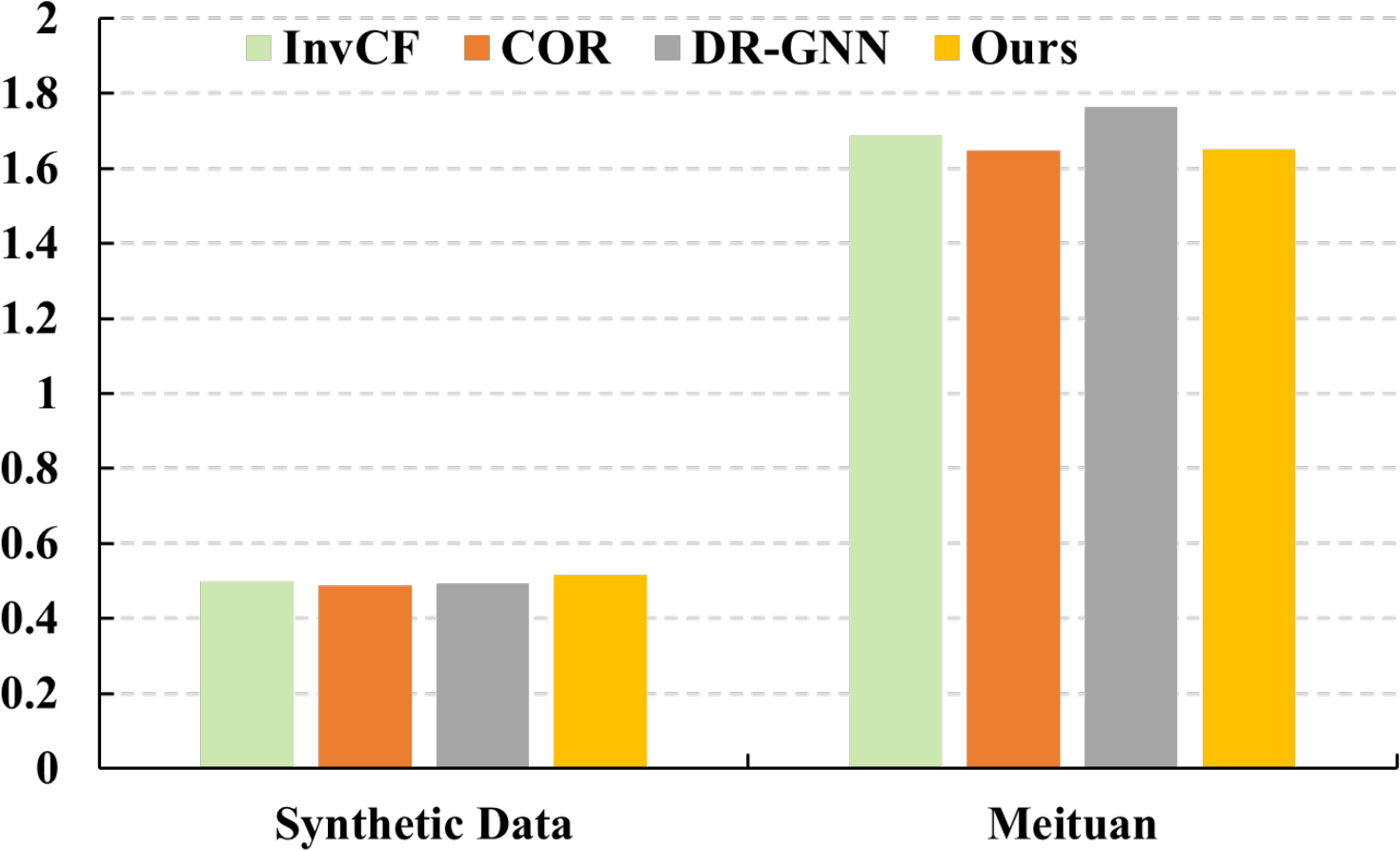}}
\vspace{10pt}
\end{minipage}
\begin{minipage}{0.45\linewidth}
\centerline{\includegraphics[width=1\textwidth]{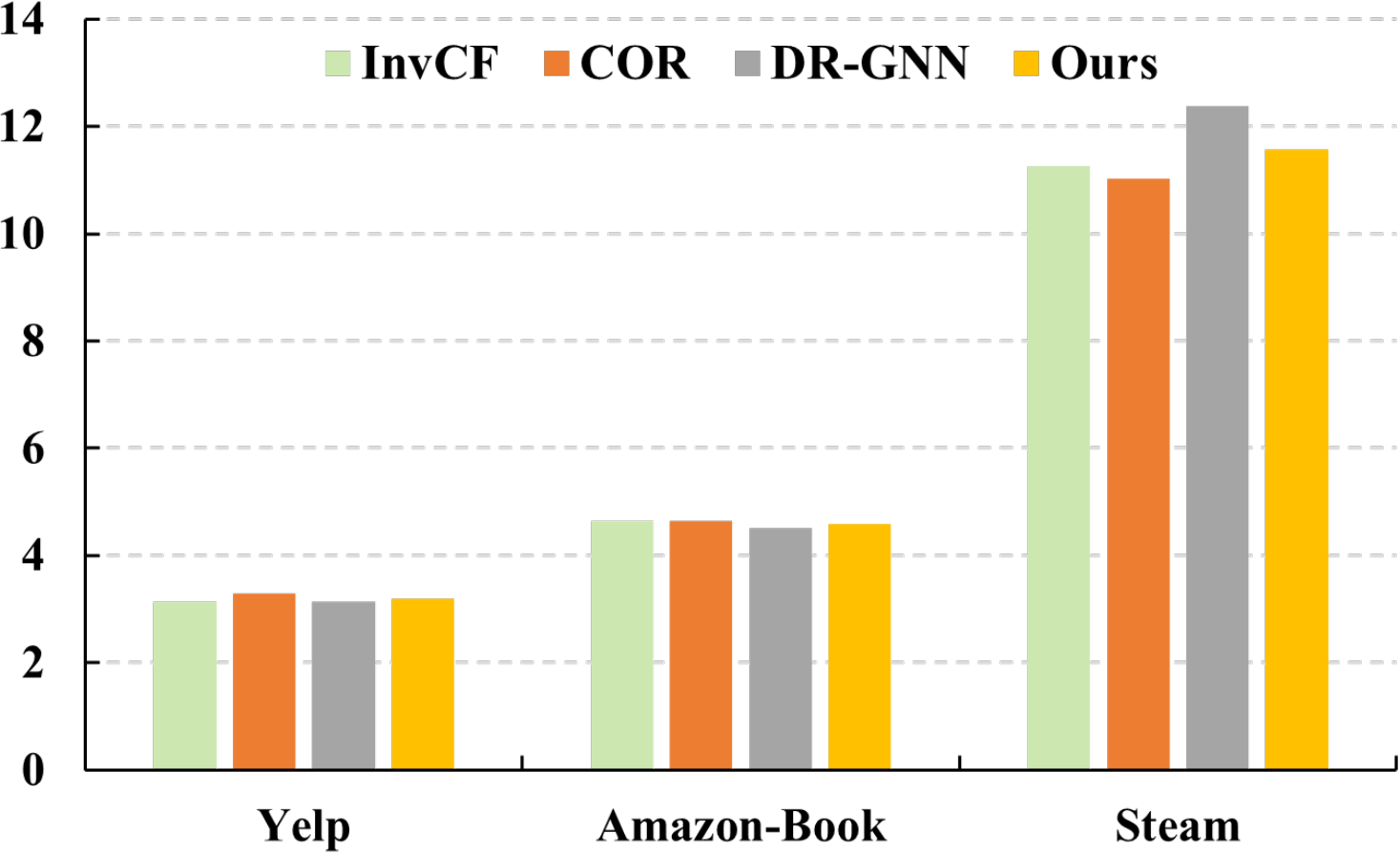}}
\vspace{10pt}
\centerline{\includegraphics[width=1\textwidth]{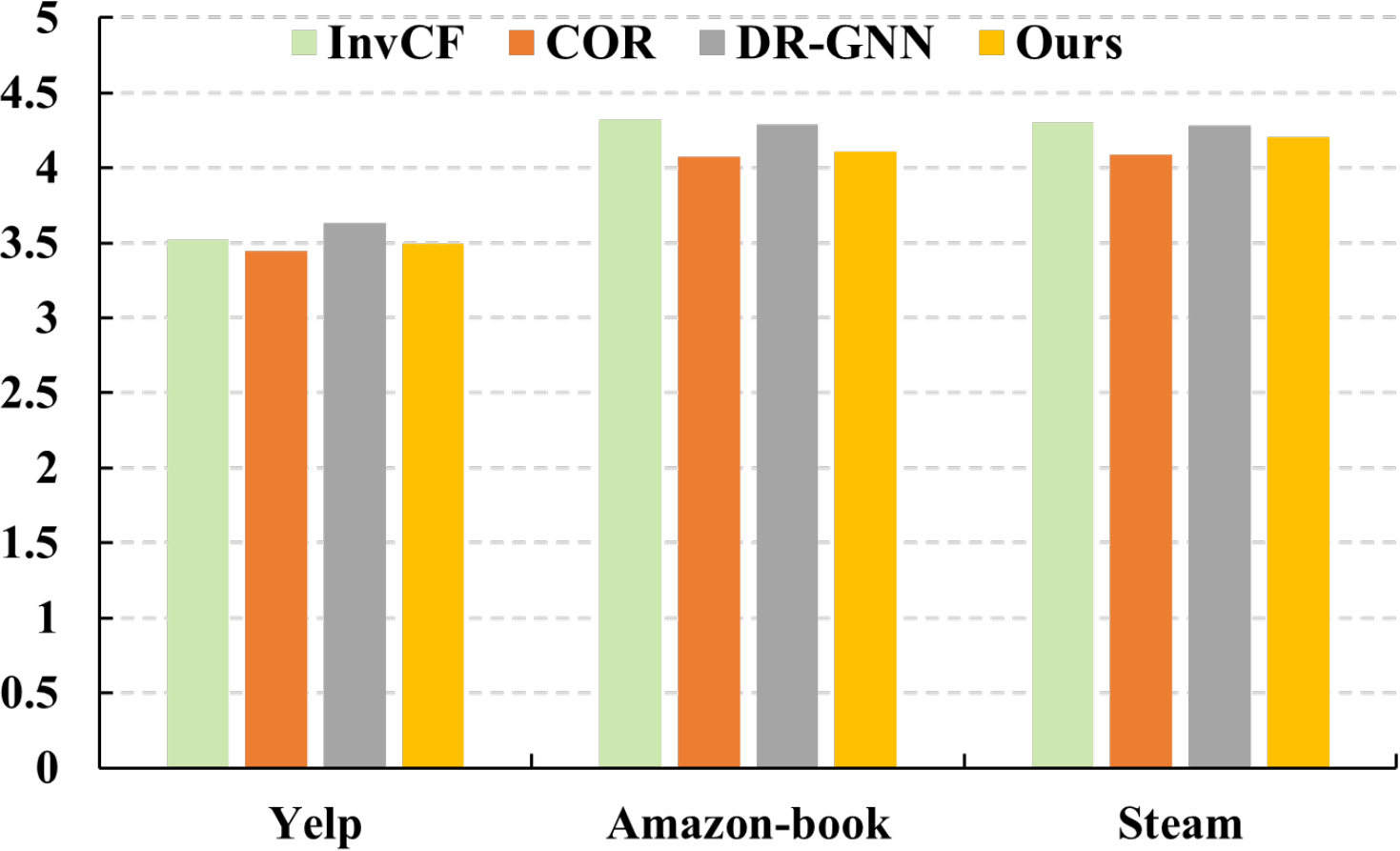}}
\vspace{10pt}
\end{minipage}
\caption{{Time cost and space cost on five datasets. Time and space complex is measured by seconds and gigabytes (GB), respectively. All results are the average value over 10 runs. The first row and second row correspond to time cost and space cost, respectively.}}
\label{time_space}
\end{figure}

\subsection{{(\textbf{RQ6})~Stability Analysis}}
{We implement experiments to explore the stability of DT3OR with different scales of distribution shift on Meituan and Amazon-Book datasets. We spit the dataset to construct temporal shift with timestamp, i.e., weekday as IID, weekend as OOD. By adjusting the proportion of the iterations with timestamp, specifically by modifying the number of weekend days, we could obtain different size of distribution shift. The experimental results are demonstrated in Fig.~\ref{change_shift}. We could observe as follows.
\begin{itemize}
    \item As the proportion of distribution shift increases, indicated by the number of weekend days, the performance of all methods deteriorates. This can be attributed to the models' insufficient ability to handle large-scale distribution shifts, resulting in poor recommendation performance. 
    \item Our method consistently achieves better performance in most cases, highlighting the generalizability of the test-time training strategies.
\end{itemize}}

\begin{figure}[]
    \centering
    \begin{minipage}{0.45\linewidth}
    \centerline{\includegraphics[width=1\textwidth]{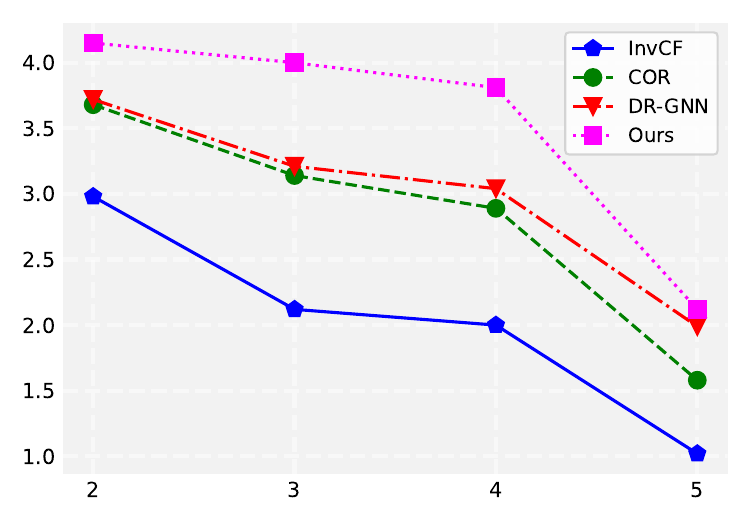}}
    \centerline{{Meituan}}
    \vspace{10pt}
    \end{minipage}
    \begin{minipage}{0.45\linewidth}
    \centerline{\includegraphics[width=1\textwidth]{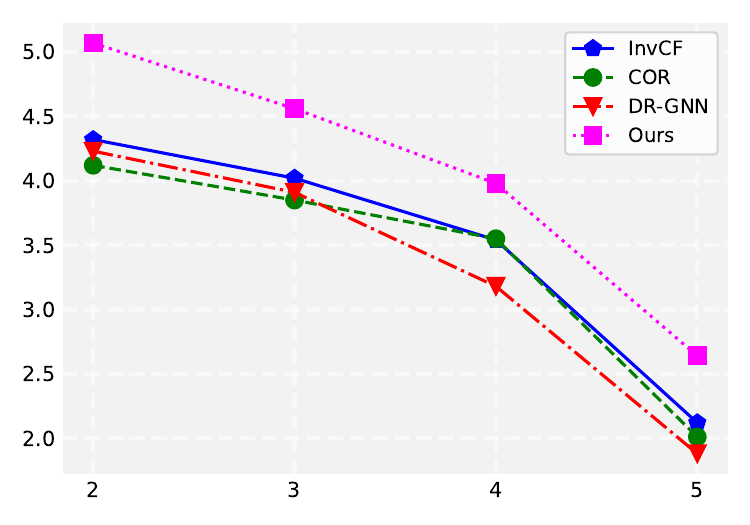}}
    \centerline{{Amazon-Book}}
    \vspace{10pt}
    \end{minipage}
    \caption{{Different size of distribution shift on Meituan and Amazon-Book datasets. The x-axis represents the number of weekend days, while the y-axis denotes the R@50 performance.}}
    \label{change_shift}
    \end{figure}

\section{Related Work}
\subsection{Test-time Training}

Test-time training (TTT)~\cite{jiaxinMM,jiaxinTKDD} is a method to improve model generalization in recent years, which allows for partial adaptation of a model based on test samples. TTT addresses the issue of distribution shifts between the training and test sets by updating the model using self-supervised learning (SSL) tasks during the test-time phase. Initially, TTT was introduced to update the model by jointly optimizing the supervised loss and SSL loss during the test period. However, these paradigms necessitate the simultaneous optimization of the supervised loss and SSL loss. To surmount this constraint, Tent~\cite{wang2020tent} presents a fully test-time training approach that solely uses test samples and a pre-trained model. Inspired by the success of TTT in handling out-of-distribution problems, we apply this technique to the OOD recommendation scenario. By leveraging TTT, we can adapt the recommendation model using test samples and a pre-trained model, improving its performance in the presence of distribution shifts.

\subsection{Self-Supervised Learning}

Self-supervised learning (SSL) algorithms have recently garnered significant attention in the domain of graph neural networks (GNNs)~\cite{huangexploring,huangnodes}. SSL is a widely employed technique for training neural networks as evidenced in various studies~\cite{CCGC,convert,Graphlearner,MGCN,liang_tkde}. A fundamental technique in SSL is contrastive learning, which focuses on increasing the similarity of positive pairs while decreasing the similarity of negative pairs. This approach has achieved remarkable success in graph learning~\cite{yujie1,yujie2}. For instance, DCRN~\cite{DCRN} strengthens feature decorrelation by aligning the representations of identical nodes across varied views while distancing the representations of dissimilar nodes in different views. Similarly, GMI~\cite{GMI} extends DGI~\cite{DGI} by incorporating edges and node features, to mitigate representation collapse. Additionally, the efficacy of amplifying mutual information to extract and learn embeddings for graph classification tasks has been demonstrated by MVGRL~\cite{MVGRL} and InfoGraph~\cite{InfoGraph}. Building on these principles, GraphCL~\cite{GraphCL} produces two augmented views for SSL, and learn the representations by aligning the same nodes in both views and separating distinct nodes.

\subsection{Recommender System}

Recommender systems have become a vital tool for discovering user's interests and preferences in various online platforms, which assume that the user and item follow the independent and identically distributed (I.I.D.). OOD recommendation has received limited attention, which can be roughly divided into three classes. 1) Disentangled recommendation approaches~\cite{dis_rec} aim to learn factorized representations of user preferences to enhance robustness against distribution shifts. 2) Causal-based methods~\cite{cor,causpref} leverage causal learning techniques to address OOD problems. However, these methods often require interventions during the training process, which may be impractical when only a pre-trained model is available. 3) Model retraining~\cite{retraining} focuses on adapting the model to the OOD environment. This approach necessitates a trade-off between retraining frequency and computational cost. Benefiting from the success of test-time training (TTT), we introduce the TTT strategy to assist model adaptation in OOD scenarios in this paper.

\section{Conclusion}

In this paper, we first introduce the test-time training strategy to address the out-of-distribution recommendation problem. The model could adapt the shift user/item feature distribution in test-time phase with our carefully designed self-supervised learning tasks. Specifically, the proposed self-distillation task and the contrastive task assist the model learn user's invariant interest performance and the shift feature distributions in the test-time phase, resulting in the model achieving better performance on test data. Moreover, we give theoretical analysis to support the rationale of our dual test-time training strategy. Comprehensive experiments have affirmed the efficacy of DT3OR.

In this paper, we make an initial attempt to alleviate the out-of-distribution recommendation problem using a test-time training strategy. The key idea is to design an unsupervised self-distillation task and a contrastive learning task for the test-time model adaptation. The core operation in these tasks is obtaining the users' preference centers through clustering. How to obtain high-quality clustering of preference centers when the dataset contains many user and item types is a valuable research question. Moreover, determining the optimal number of test-time training iterations for effective parameter updates remains a challenging issue. Excessive iterations may cause the model to lose valuable information learned during the initial training phase, while insufficient iterations may prevent the model from fully capturing the shifted data distribution. In the future, we will further explore the aforementioned issues to enhance the generalization ability of recommendation models in OOD scenarios.

\section*{Acknowledgments}
This work was supported by the National Science and Technology Innovation 2030 Major Project under Grant No. 2022ZD0209103, the National Natural Science Foundation of China (project No. 62325604, 62276271, 62406336, 62376039, U24A20253, 62476281), and the Program of China Scholarship Council (No. 202406110009).

\section{Appendix}
\subsection{Hyper-parameter Settings}\label{hyper}
\begin{table}[]
\centering
\caption{Hyper-parameter setting for our method.}
\scalebox{0.85}{
\begin{tabular}{cccccc}
\hline
{\textbf{Dataset}} & {\textbf{Synthetic data}} & {\textbf{Meituan}} & {\textbf{Yelp}} & {\textbf{Amazon-Book}} & {\textbf{Steam}} \\ \hline
{Learning Rate}    & {1e-4}                    & {1e-5}             & {1e-4}          & {1e-4}                 & {1e-4}           \\
{Drop Rate}        & {0.5}                     & {0.5}              & {0.5}           & {0.5}                  & {0.5}            \\
{Epoch}            & {10}                      & {10}               & {10}            & {10}                   & {10}             \\
{$\alpha$}           & {1.0}                     & {1.0}              & {1.0}           & {1.0}                  & {1.0}            \\
{$\tau$}             & {0.9}                     & {0.8}              & {0.9}           & {0.9}                  & {0.9}            \\
{$T$}    & {0.1}                    & {0.3}             & {0.3}          & {0.1}                 & {0.1}           \\
{$K$}                & {4}                       & {5}                & {4}             & {4}                    & {5}              \\ \hline
\end{tabular}}
\label{hyper_para}
\end{table}
{In this subsection, we report the statistics summary and hyper-parameter settings of our proposed method in Tab.~\ref{hyper_para}.}

\subsection{The Construction for Synthetic Data}\label{Constrcution}
We elaborate details on the process of constructing the Synthetic dataset in this part. Following the methodology in COR~\cite{cor}, we constructed the dataset through four primary steps:

1) Sampling user and item features;

2) Estimating user preferences;

3) Sampling user interactions;

4) Collecting OOD (Out-of-Distribution) data.

Specifically, the first step involves sampling user and item features from the standard Gaussian distributions, denoted as $\mathcal{N}(0,1)$ for users and $\mathcal{N}(-1,1)$ for items. We set the number of users and items to be 1,000 each. The second step pertains to user preference estimation, which is grounded in prior knowledge. For instance, higher income might positively influence preferences for more expensive items. The relationships can be either positive or negative. We calculate user preferences by summing the user features with either positive or negative weights and then apply a sigmoid function to introduce non-linearity into these relationships.

In the third step, sampling of user interactions is carried out. Utilizing the acquired user preferences and features of item, we ascertain the relevance of user-item interactions via the sigmoid function. Interaction data is subsequently sampled from the Bernoulli distribution~\cite{bern}. The final step concerns the collection of OOD data. In this regard, we re-sample user characteristics from $\mathcal{N}(1,1)$ to accumulate user interactions in an OOD environment, maintaining item features constant. Steps 2 and 3 are then reiterated to sample new user interactions under these altered conditions. This rigorous process ensures a comprehensive and well-structured Synthetic dataset for our experiments.



\bibliographystyle{IEEEtran}
\bibliography{myreference}

\end{document}